%% file: CBSN_draft_3rdRR_MSS_v01.tex
\documentclass[12pt,reqno]{amsart}

\usepackage{amsmath}
\usepackage[foot]{amsaddr}
\usepackage{amssymb}

\usepackage{amsthm}

\usepackage{bm}

\usepackage{bbm}
\usepackage{graphicx}
\usepackage{xkeyval}

\usepackage{libertine}

\usepackage{libertinus}
\usepackage[T1]{fontenc}

\usepackage[final]{microtype}
\usepackage{caption}
\usepackage{subcaption}

\usepackage[space]{grffile}
\usepackage{import}

\usepackage[dvipsnames]{xcolor}
\usepackage[hidelinks, colorlinks, urlcolor=BrickRed,linkcolor=black, citecolor=BrickRed]{hyperref}

\usepackage[nodisplayskipstretch]{setspace}
\usepackage{pdflscape}
\usepackage{float}
\usepackage[margin=1in]{geometry}
\usepackage[longnamesfirst,authoryear]{natbib}
\usepackage{multirow}
\usepackage{array}
\usepackage{dcolumn}
\usepackage{booktabs}
\usepackage{tablefootnote}
\usepackage{tikz}

\usepackage{dsfont}

\DeclareMathOperator*{\argmax}{arg\,max}

\let\hbar\relax

\newtheorem{lemma}{Lemma}
\newtheorem{proposition}{Proposition}
\newtheorem{corollary}{Corollary}
\newtheorem{example}{Example}
\newtheorem{definition}{Definition}

\newtheorem{result}{Result}

\begin{document}

\author{Marcos R. Fernandes$^{ (1)}$}
\address[1]{University of São Paulo, Department of Economics, e-mail: \textsf{mrf.ross@gmail.com}.}
\dedicatory{This version: February 2023}

\begin{abstract}
In this study, I present a theoretical social learning model to investigate how confirmation bias affects opinions when agents exchange information over a social network. Hence, besides exchanging opinions with friends, agents observe a public sequence of potentially ambiguous signals and interpret it according to a rule that includes confirmation bias. First, this study shows that regardless of level of ambiguity both for people or networked society, only two types of opinions can be formed, and both are biased. However, one opinion type is less biased than the other depending on the state of the world. The size of both biases depends on the ambiguity level and relative magnitude of the state and confirmation biases. Hence, long-run learning is not attained even when people impartially interpret ambiguity. Finally, analytically confirming the probability of emergence of the less-biased consensus when people are connected and have different priors is difficult. Hence, I used simulations to analyze its determinants and found three main results: i) some network topologies are more conducive to consensus efficiency, ii) some degree of partisanship enhances consensus efficiency even under confirmation bias and iii) open-mindedness (i.e. when partisans agree to exchange opinions with opposing partisans) might inhibit efficiency in some cases.
\\


\vspace{2mm}

\end{abstract}
\title{Confirmation bias in social networks}

\thanks{I thank Sandro Brusco, Pradeep Dubey, Ben Golub, Laura Karpuska, Ehud Lehrer, Ting Liu, Mihai Manea, Alejandro Melo, Ron Peretz, Max Silva, Michael Sinkey, Eilon Solan, Troy Tassier, Yair Tauman, and two anonymous referees for helpful comments and suggestions. I also thank all participants of the 29th International Conference on Game Theory (Stony Brook), 23rd Coalition Theory Network workshop (Maastricht), 4th Annual Conference on Network Science and Economics (Nashville), 44th Eastern Economic Association conference (Boston) and the seminar participants at UFABC, Fordham University, Stony Brook Center for Behavioral Political Economy, INSPER and University of São Paulo. I gratefully acknowledge the Department of Economics of Stony Brook University (Juan Carlos Conesa), Stony Brook Graduate Students Organization and the 4th Network Science in Economics conference series (Myrna Wooders) for financial support.}

\maketitle

\section{Introduction}
\label{sec:intro}

People form opinions on various economic, political, social, and health issues based on information from both the media and people they trust (e.g. friends, coworkers, family, experts, etc). This information acquisition process usually occurs when the issue discussed has no clear-cut \textit{right/wrong} or \textit{true/false} distinction or when the available information cannot be easily understood. Consulting people's opinions, in this case, is an appealing and easy way to gather information. For many people, social networks then become primary tool to stay informed. Thus, understanding how beliefs depend on how agents perceive and process information is vital. In this study, I examine how opinions are affected by confirmation bias in a networked environment.

In psychology, confirmation bias denotes the interpretation of evidence in ways consistent with existing beliefs (\cite{Nickerson1998}, \cite{molden2008}). This can be done in different ways, like restricting attention to favored hypothesis, disregarding evidence that could falsify the current worldview or overvaluing positive confirmatory instances. In all cases, people restrict attention to a single hypothesis and fail to carefully consider alternatives.

In social psychology, people interpret evidence when they are ambiguous (i.e. when evidence is conflicting). People may misinterpret scientists and experts after ambiguous announcements. \cite{simonovic2022psychological} highlighted that when WHO declared the COVID-19 outbreak a global pandemic in 2020, experts did not precisely understand the extent and nature of the health risks or how disease transmission can be prevented. Hence, WHO provided \emph{conflicting recommendations} to the public on whether wearing a mask was necessary. Other medical authorities also provided conflicting recommendations to the public regarding medicines and vaccines' efficacy. Conflicting evidences may have even contributed to people making their own assessment about the problem.

While friends may help people to aggregate information in some cases, in other cases, people may expose themselves to others who that rely on their own worldview to derive information from ambiguous evidence. In these cases, efficient aggregation of information is not guaranteed, and I investigate how opinions are influenced by people's biases.

To analyze this phenomenon, I consider a society where agents are interested to learn the underlying state $\theta\in \Theta = [0,1]$. For instance, the underlying state $\theta$ might represent the efficacy of a new vaccine (e.g. from 0 to 1). All agents have prior beliefs about the vaccine's efficacy and observe a sequence of public signals, one at each date $t$. Public signals may be (i) informative or (ii) ambiguous. Informative signals are binary variables indicated as $1$ if state on the right side of the 0-1 spectrum are more likely (i.e. if vaccine's efficacy is high) and $0$ if the states on the left side of the 0-1 spectrum are more likely (i.e. if vaccine's efficacy is low). Hence, as signals realization does not convey full information on the underlying state, agents can only learn the true state (vaccine's efficacy) asymptotically. This is in the spirit of ongoing learning, where information accumulates through experience. In the case of ambiguous signals, agents are allowed to interpret these signals using a fairly general randomization rule proposed by \cite{fhj2018} that accounts for confirmation bias. Hence, the interpretation of the ambiguous signal received at time $t$ is influenced, to a greater or lesser extent, by the likelihood of 0 and 1 at time $t-1$ (see more details below). This captures situations wherein people feel impelled to explain ambiguous evidence about a particular issue.

As in \cite{jmst2012}, besides learning from public signals, agents exchange information through a social network. At the beginning of every period $t$, the public signal is realized. Thus, each agent first \textit{interprets} signals (if ambiguous) using the randomization rule, \textit{stores} the signal and \textit{computes} the Bayesian posterior (opinion and precision). Every agent then sets their \textit{final} opinions and precisions to be a linear combination of the Bayesian posterior opinions and precisions computed with the interpreted signal and opinions and precisions of friends (e.g., formal definition of neighbors in subsection \ref{subsec:netstr}) they met in the period before. Social connectivity among agents remains fixed over time and strong connectivity is assumed (i.e., all agents are exposed to all other agents either through a directed or undirected path in the social networks). 

Hence, despite the level of ambiguity and both in the case of a single individual or a connected society, only two types of opinions can emerge, and both are biased: left- and right-biased opinions. However, one type of opinion is less biased than the other depending on the underlying state. Less-biased opinion is only guaranteed to emerge under a favorable combination of sufficiently low ambiguity and sufficiently pronounced states. If this condition holds, I show that the less-biased opinion is attained with probability 1. Moreover, long-run learning is not attained even if people are impartial when they interpret ambiguous signals (i.e., when interpreting evidence uniformly at random instead of using their own opinions). Those results contrast with those by \cite{Rabin1999} and \cite{fhj2018}, who suggest that long-run learning occurs with a positive probability and that impartiality helps in learning the state. Furthermore, both the network effect presented here and signals realization, reinforce the interpreting dispute (\textit{tug-of-war}) as people may have their own interpretation biases reinforced or attenuated by other agents.

Finally, confirming the probability of emergence of the less-biased consensus analytically is difficult, and I use Monte Carlo simulations to show its determinants. The presence of partisan agents (i.e., agents with skewed initial priors) in societies suffering from confirmatory bias have two main effects. (i) When the degree of partisanship is low, partisanship helps to counter the realization of initial misleading signals (e.g., realization of a 0 when $\theta\geq0.5)$. Thus, low partisanship increases the odds of reaching less-biased consensus. (ii) When the degree of partisanship is high, partisans exacerbate misinterpretation of signals. Thus, high partisanship reduces the odds of reaching less-biased consensus. Moreover, I also show that open-mindedness of partisan agents (i.e., when partisans agree to exchange opinions with partisans with polar opposite beliefs) might reduce the odds of reaching less-biased consensus in some network structures.
 
While this work does not generalize theoretical results for other conjugate families and numerical results for other network structures, both methods and cases explored are sufficiently general to capture important aspects of real-world networks. In every period, public signals realized and observed by all agents may represent information reported by sources including media outlets and international organizations. The level of ambiguity of the informational content reported by them, measured by a parameter $\mu \in (0,1)$, represents the fraction of instances where a signal simultaneously conveys two conflicting meanings and agents feel impelled to interpret them. Parameters of the signal interpretation function dictate the interpretation behavior of every agent.

This work is structured as follows. Section 2 provides a brief literature review and highlights contributions. Section 3 describes a framework for updating beliefs when agents communicate over social networks with ambiguous signals and present main theoretical results. Section 4 describes a simulation exercise when priors heterogeneity (partisanship) is assumed. Section 5 concludes the study. Moreover, six appendices are available. Appendices A and B contain primitives of the Beta-Bernoulli conjugate family employed in this work. Appendix C contains proofs of auxiliary results, while Appendix D presents proofs of main results. Appendices E and F show simulation statistics and present regression robustness.

\section{Literature review and contribution}
\label{sec:litrev}

Considerable empirical evidence on social psychology supports the idea that confirmation bias is extensive and appears in many ways. Most studies in the field confirm the human tendency of casting doubt on information that conflicts with preexisting beliefs and confirming preexisting beliefs when exposed to ambiguous information (see \cite{Nickerson1998}). However, this selectivity in the acquisition and use of evidence occurs without intending to treat evidence in a biased way. \cite{Molden2004,molden2008} note that both \textit{vagueness} (when evidence is weak) and \textit{ambiguity} (when evidence is conflicting and open to interpretations) induce interpretation. Conversely, \cite{Furnham1995} and \cite{furnham2013} review literature on the subject and report evidence that the way people perceive and process information about ambiguous situations is related to their degree of ambiguity tolerance (i.e., individual differences in cognitive reaction to stimuli considered ambiguous). Therefore, ambiguity tolerance refers to underlying psychological differences that impel people to process, interpret, and react differently to ambiguous information.\footnote{More recently, scholars have considered the concept of tolerance of ambiguity as a reflection of the contemporary definition of ambiguity proposed by \cite{ellsberg1961}. For a good coverage of the classic literature on ambiguity aversion, see \cite{Gilboa1989}, \cite{Gilboa1993}, and \cite{Epstein2007}.}

Thus, confirmation bias may oppose standard Bayesian updating processes as agents scrutinize signals in line with their worldviews. Some examples of decision-making models that account for Bayesian updating deviation are \cite{Hellman1970}, \cite{Rabin1999}, \cite{Wilson2014}, \cite{fhj2018}, \cite{sikder2020minimalistic} and \cite{buechel2022misinformation}. 

Studies by \cite{Rabin1999}, \cite{fhj2018}, \cite{sikder2020minimalistic} and \cite{buechel2022misinformation} are the closest references to this work, in both spirit and results. \cite{Rabin1999}, showed that signals believed as less likely are misinterpreted with an exogenous probability. \cite{fhj2018} stated that ambiguous signals, in its simplest version with binary states, are produced with a certain probability, and agents interpret those before conducting the Bayesian update. Interpreting these signals requires agents to use three methods that differ in intensity with which agents conform their interpretation with their current worldview. However, both works do not consider network communication.

\cite{sikder2020minimalistic} employ a slightly modified version of \cite{Rabin1999} to a networked environment (mostly focused on regular networks), where agents synchronously share the full set of signals with their neighbors. However, biased agents reject information incongruent with their preexisting beliefs, reduce the weight they place on other agents, and place the remaining weight on an external positively oriented "ghost" node, creating a polarization of unbiased agents in the steady state. In this study, I assume a general (connected) network structure among agents and allow them to set their final beliefs to be a linear combination of the Bayesian posterior and opinions of their neighbors as in \cite{jmst2012}, regardless of their biases and signals received. A key difference relative to \cite{sikder2020minimalistic} is that my modeling strategy allows me to discuss the relative importance of the learning parameters, network structure, and connections heterophily (open-mindedness of partisans) in determining the probability of reaching the less-biased consensus.

\cite{buechel2022misinformation} allow different types of signals to have different transmission capacities (i.e., asymmetric decay factor applied to positive and negative signals) when signals are shared in different networks, and show that for a society to aggregate information efficiently, different asymmetries must be balanced and that an agent's ratio of centralities between the two networks must be moderate compared to the ratio of centrality concentration in the two networks. In my model, sharing asymmetries and misinformation is not considered even when partisan agents remain equally balanced and central. However, this work shares an interesting feature with my model relative to how equality of centralities is critical for reducing misinformation.



This work is also related to the literature of \textit{bounded confidence} in networks. Overall, this literature focuses on models of social learning wherein agents overvalue the opinion of friends with similar beliefs. \cite{HegKrau2002}, \cite{HegKrau2005}, \cite{Dandekar2013}, \cite{Mao2018}, and \cite{gallo2020social} provide examples of this phenomenon. While bounded confidence involves the tendency to \textit{conform} with the majority or leading people, confirmation bias is a failure in the Bayesian updating process. From this perspective, modeling confirmatory bias as either a bounded confidence or a failure in the Bayesian update has different consequences. On the one hand, bounded confidence presumes that the connections between agents are broken (or temporarily interrupted) according to opinions distance. Hence, nontrivial changes are implied in the network topology. In this literature, long-run polarization naturally occurs under bounded confidence. Polarization, hence, is a natural product of the initial heterophily of opinions in the system and eventual deletion of links. On the other hand, modeling confirmatory bias as a Bayesian update failure is inconsequential to the network topology and under the strong connectivity assumption leads to a bias (misinformation) that can be analytically studied.

Finally, numerous works on \textit{social learning} assumed bounded and full rationality. Bayesian social learning literature (fully rational agents) mainly focuses on  formulating stylized games with incomplete information and characterizing its equilibria. Specifically, rather than considering complex and repeated interactions, most works focus on environments where agents are myopic or interact only once (\cite{Banerjee1992}, \cite{balagoyal1998}, \cite{Bala2001}, \cite{bf2004}, \cite{adlo2011}).\footnote{For an overview of recent research on learning in social networks, see \cite{Acemoglu2011,golub2017learning,grabisch2020survey}.}

Non-bayesian learning (bounded rational agents) literature focuses on studying generalizations of the seminal \cite{DeGroot1974} model. \cite{DeMarzo2003} show that consensus result does not rely on the social weighting matrix being a stationary matrix. \cite{aceozdpar2010} consider a random meeting (Poisson) model and characterize how the presence of forceful agents (i.e., agents who influence others disproportionately and hardly revise their beliefs) prevents information aggregation. Conversely, \cite{Golub2010} show that convergence holds if (and only if) the influence of the most influential agent vanishes as society grows unboundedly. \cite{jmst2012} is the first study to consider the possibility of constant arrival of informative signals in every period in networked environments. In their study, the update rule that sets the final belief as a linear combination of the Bayesian posterior and the  neighbors' opinions is an efficient alternative to the complicated task of implementing Bayesian update in networks. Finally, similar to this work in modeling strategy, \cite{AzzFer2018} investigate how the structure of social networks and the presence of fake news affect the degree of polarization and misinformation. In their study, i) their model considers the presence of \textit{bots} whose sole purpose is to deceive other agents, and that ii) connectivity among all agents evolves stochastically. Those two features combined are main drivers of misinformation and polarization cycles. Herein, the main source of bias derives from confirmation bias and connectivity among agents is assumed fixed. Thus, this study focuses on understanding how misinformation depends on both network structure and how agents interpret ambiguous signals.

\section{The model}
\label{sec:model}

\subsection*{Notation:} All vectors are considered column vectors, unless stated otherwise. Given a vector $v \in \mathbb{R}^{n}$, I denote by $v_i$ its $i$-th entry. When $v_i \geq 0$ for all entries, I write $v\geq0$. Moreover, I define $v^{\top}$ as the transpose of the vector $x$. Hence, the inner (scalar) product of two vectors $x,y \in \mathbb{R}^{n}$ is denoted by $x^{\top}y$. I denote by $\mathbf{1}$ the vector with all entries equal to 1. A matrix $W$ is considered to have size $m \times n$ whenever $W$ has exactly $m$ rows and $n$ columns. Moreover, whenever $m = n$, $W$ is called a square matrix of size $n$. The identity matrix of size $n$ is denoted by $\mathbb{I}$. For a matrix $W$, $W_{ij}$ denotes the entry in the $i$-th row and $j$-th column. The notation $W_{ij}^{k}$ is used to denote the entry in the $i$-th row and $j$-th column of the matrix $W^{k}$, i.e. the matrix $W$ raised to the power $k$. Finally, a vector $v$ is said to be a stochastic vector when $v\geq0$ and $\sum_{i}v_i = 1$. A square matrix $W$ is said to be a (row) stochastic matrix when each row of $W$ is a stochastic vector.


\subsection{Network structure}
\label{subsec:netstr}
Connectivity among agents in a network is described by a {directed graph} $G = \left(N, g\right)$, where $N = \{1,2,\dots,n\}$ is the set of agents, fixed over time, and $g$ is a binary $n\, \times\, n$ adjacency (or incidence) matrix, also fixed over time. Each element $g_{ij}$ in the directed-graph represents the connection between agents $i$ and $j$. More precisely, $g_{ij} = 1$ if individual $i$ is paying attention to (i.e. receiving information from) individual $j$, and 0 if otherwise. As the graph is directed, some agents pay attention to others who are not necessarily reciprocating (i.e., $g_{ij} \neq g_{ji}$). The {out-neighborhood} of any agent $i$ is the set of agents that $i$ is receiving information from, and is denoted by $N_{i}^{out} = \{j \, | \, g_{ij} =1\}$. Similarly, the {in-neighborhood} of any agent $i$ is denoted by $N_{i}^{in} = \{j \, | \, g_{ji} =1\}$, represents the set of agents that are receiving information from $i$. In undirected networks, $N_{i}^{in} = N_{i}^{out} = d_i$, where $d_i$ is the number of neighbors agent $i$ has, also known as degree centrality of agent $i$. Thus, the term $\hat{g}_{ij}=\frac{g_{ij}}{|N_{i}^{out}|} \in [0,1]$ represents the weight that agent $i$ gives to information received from their out-neighbor $j$. A network is considered \emph{regular} if every node has the same degree of centrality, and that a network is \emph{complete} if every node is connected with all other nodes. Finally, a directed path in $G$ from agent $i$ to agent $j$ is defined as a sequence of agents beginning with $i$ and ending with $j$ such that each agent is a neighbor of the next agent in the sequence. A social network is \textit{strongly connected} if a directed path from each agent to any other agent exists.

\subsection{Signals, initial beliefs and opinions.}
\label{subsec:belsig}
Let $\Theta = [0,1]$ to denote the set of possible states of the world. For instance, one may find useful to interpret $\Theta$ as the effectiveness of a new vaccine, such that a state close to 0 means that the vaccine has low efficacy, whereas a state close to 1 means that vaccine has high efficacy. 

Conditional on the state of the world $\theta$, every agent observes a sequence of public signals $s_t$, one at each date $t \in \{1,2,\dots\}$. Public signals lie in the set $S = \{1, 0, a\}$. Considering the example of the vaccine's efficacy given above, a signal $1$ is evidence that the new vaccine can prevent people from severe illness, a signal $0$ is evidence of no efficacy, and a signal $a$ is \textit{ambiguous} and open to idiosyncratic interpretation (Section~\ref{subsec:signals} explains how agents deal with those signals). Signals are independent over time, conditional on the state. The probability that a signal is ambiguous is $\mu \in (0,1)$. Hence, the signal conveys informational aspects that could lead one to interpret as either $1$ or $0$. With the remaining probability ($1-\mu$), the information provided by the signal is unambiguous. In any state $\theta \in \Theta$, the probability that an unambiguous signal is 1 is $\theta \in [0,1]$ and 0 with probability $1-\theta$. The signal structure is depicted in the Figure~\ref{fig:signals}.

\begin{figure}[H]
\centering
\includegraphics[page=2,scale=0.6, trim = 11cm 5cm 11cm 5cm, clip]{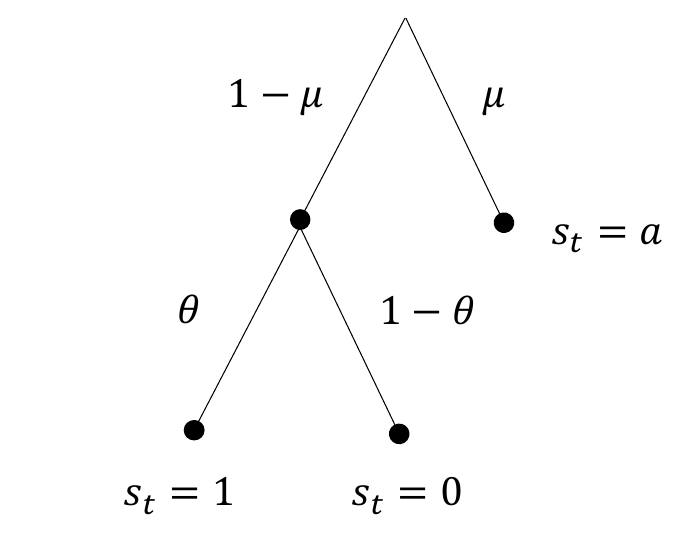}
\caption{Signals structure}
\label{fig:signals}
\end{figure}

Each agent $i$ in this society is assumed to start with an \textit{initial belief} about an underlying state $f_{i,0}(\theta) \in \Delta\Theta$, represented by a Beta probability distribution over the set $\Theta$ with shape parameters $\alpha_{i,0}, \beta_{i,0} \geq 1$ and defined as follows:

\begin{equation}
f_{i,0}\left(\theta\right) = 
\begin{cases}
\dfrac{\Gamma\left(\alpha_{i,0}+\beta_{i,0}\right)}{\Gamma\left(\alpha_{i,0}\right)\Gamma\left(\beta_{i,0}\right)} \, \theta^{\alpha_{i,0}-1} (1-\theta)^{\beta_{i,0}-1} & \text{, for } \, 0<\theta<1 \\
0 & \text{, otherwise,}
\end{cases}
\label{eq:betapdf}
\end{equation} where $\Gamma(\cdot) $ is a Gamma function and the ratio of Gamma functions in the expression above is a normalization constant that ensures that the total probability integrates to 1.

Given prior beliefs and signals, \textit{opinion} of agent $i$ at time $t$ is denoted by 
\[
y_{i,t}=\mathbb{E}\left[\theta | s_{i}^{t}\right]=\frac{\alpha_{i,t}}{\alpha_{i,t}+\beta_{i,t}},
\] where $s_{i}^{t}$ is the history of signals received and interpreted by agent $i$ up until time $t$.\footnote{Appendix A discusses the primitives of the Beta distribution and the Beta-Bernoulli conjugate family. For tractability, the opinion is intended as a real number that summarizes the entire belief. Hence, one can understand the opinion of an agent as the Bayesian estimator of $\theta$ that minimizes the mean squared error. One could also assume that the opinion of any agent $i$ at time $t$ could also be the Bayesian estimator of $\theta$ which minimizes the absolute error. As the mean, mode, and median of the Beta distribution are asymptotically equivalent, the functional form is irrelevant for the results.} The parameters' update rule will be described with more details in section 3.4.


\subsection{Interpretation of ambiguous signals}
\label{subsec:signals}

Although ambiguous signals are uninformative about the state and should be disregarded from a pure Bayesian perspective, agents are constrained to interpret ambiguous signals. This constraint captures the idea that, in some instances, people react to ambiguous pieces of information. They fail to perceive the lack of informational content of signals and end up using their prior worldview to derive meaning from them. 

For the interpretation of ambiguous signals, I use a randomization rule proposed in \cite{fhj2018}, adapted here for some technical idiosyncrasies. Hence, with probability $\gamma_{i} \in [\frac{1}{2},1]$ agent $i$ conforms with his posterior at time $t-1$ and with probability $1-\gamma_{i}$ goes against it. Essentially, with probability
\begin{equation}
\label{eq:randomization}
\psi_{i,t} = \gamma_{i} \, \mathbbm{1}\{y_{i,t-1}\geq0.5\} + (1-\gamma_{i}) \, \mathbbm{1}\{y_{i,t-1}<0.5\}
\end{equation} agent $i$ interprets the ambiguous signals as 1 and with the remaining probability $(1-\psi_{i,t})$ interprets the ambiguous signals as 0 at time $t$.\footnote{From Appendix B, note that, since mean and mode of the Beta distribution are very close for different choices of $(\alpha,\beta)$ and are asymptotically equivalent, using $y_{i,t-1}$ (the mean and mode) to interpret public signals in Equation \eqref{eq:randomization} is neutral to all results.}

Therefore, parameter $\gamma_{i}$ represents the intensity of the confirmatory bias of an individual $i$. I only assume $\gamma_{i}$ to be independent of opinion $y_{i,t}$ for any $i\in N$, history of opinions of all agents, and of all other parameters in this model. From this randomization rule, three cases of interest are available.
\begin{definition}
An individual $i \in N$ 
\begin{enumerate}

\item is \textbf{impartial} if $\gamma_{i}=\frac{1}{2}$,

\item has \textbf{confirmatory tendency} if $\frac{1}{2} < \gamma_{i} < 1$,

\item is \textbf{fully biased}  (biased for short)  if $\gamma_{i}=1$.

\end{enumerate}
\end{definition}

Hence, the signal interpretation functions, $s_{t}^{(0)}$ and $s_{t}^{(1)}$, for each individual at any point in time can be generally defined as follows:
\begin{align}
s_{i,t}^{(0)} &= \mathbbm{1}\{s_{t}=0\} +  \mathbbm{1}\{s_{t}=a\}\mathbbm{1}\{u_{t} > \psi_{i,t}\} \label{eq:signalinterpret0} \\
s_{i,t}^{(1)} &= \mathbbm{1}\{s_{t}=1\} +  \mathbbm{1}\{s_{t}=a\}\mathbbm{1}\{u_{t} \leq \psi_{i,t}\}, \label{eq:signalinterpret1}
\end{align} where $\psi_{i,t}$ is as defined in Equation \eqref{eq:randomization}, $s_{t}$ is the publicly observed signal and $u_{t}$ is the realization of a continuous $U\left[0,1\right]$ random variable at time $t$ simply used to break the tie. Draws $\{u_{t}\}$ are independent across time and also independent of all other random variables in this model. Hence, the signal interpretation functions are basically transforming observed signals $\{s_{t}\}_{t=1}^{\infty}$ into binary interpretations. When the realized public signal is $s_{t} = 1$ ($s_{t} = 0$), all agents undoubtedly interpret it as 1 (as 0) and set $s_{t}^{(0)} = 0$ and $s_{t}^{(1)}=1$ (set $s_{t}^{(0)} = 1$ and $s_{t}^{(1)}=0$). However, when the realized public signal is ambiguous (i.e., $s_{t} = a$), agents use their prior information (summarized by $y_{i,t-1}$) to categorize the signal as either 0 or 1, as per Equation \eqref{eq:randomization}. Figure~\ref{fig:signals-interp} shows a more detailed description of the signals interpretation scheme. Appendix A shows details on the signals likelihood function.

\begin{figure}[H]
\centering
\includegraphics[page=1,scale=0.53, trim = 8cm 4cm 5cm 3cm, clip]{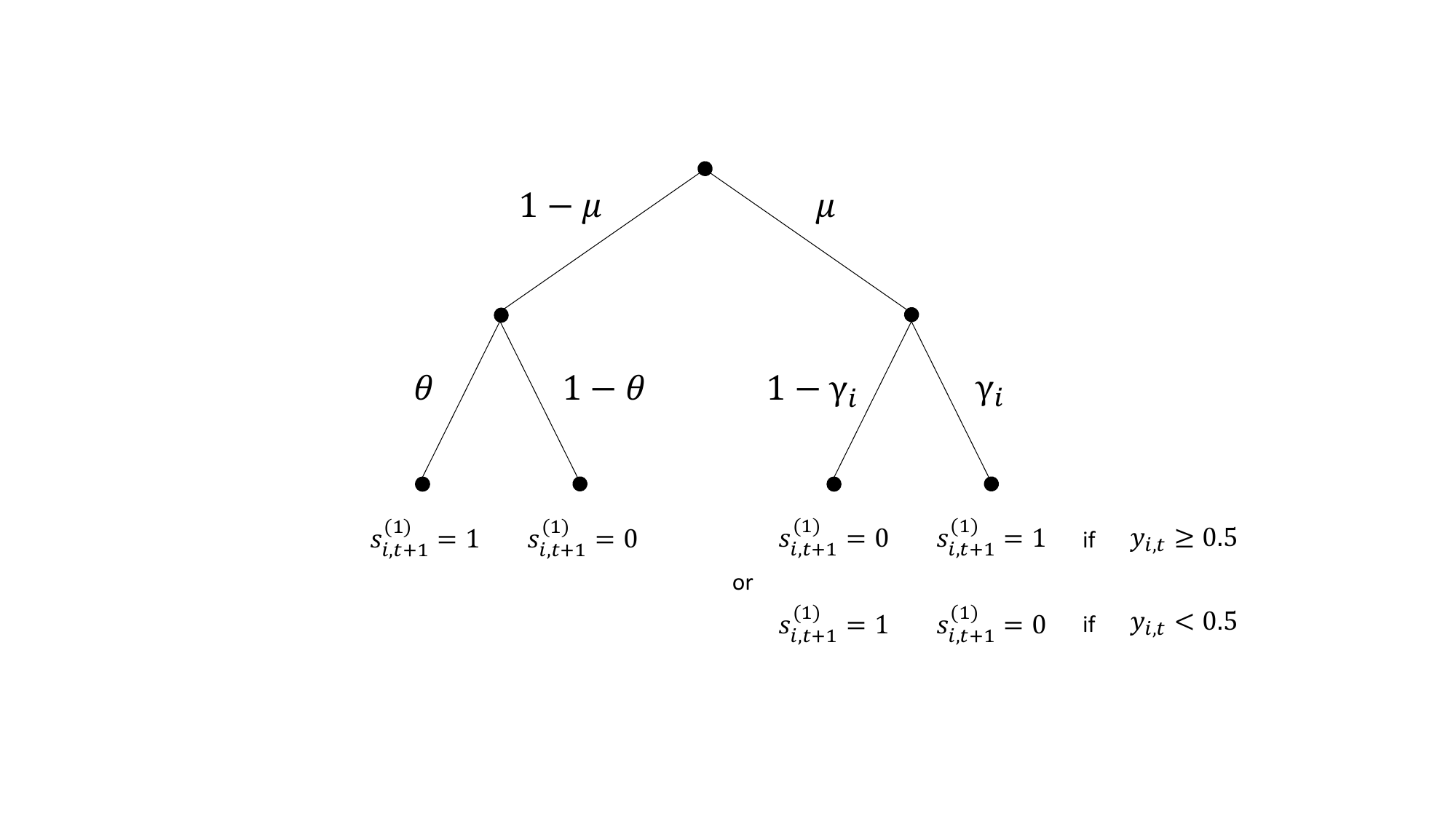}
\caption{Signals interpretation by agent $i$ upon receiving a public signal $s_{t+1}$}
\label{fig:signals-interp}
\end{figure}

\subsection{Belief evolution}
\label{subsec:belief}
Agents are assumed to update their beliefs based on public signals $s_{t}\in S=\{1,0,a\}$ and on the influence of friends in their social network. 

Hence, at the beginning of period $t$, a public signal is realized and signal $s_{t}$ is observed by agent $i$. After observing the public signal $s_{t}$, agent $i$ computes his posterior in a standard Bayesian fashion. Following \cite{jmst2012}, I assume that the updated parameters $\alpha$ and $\beta$ will be a convex combination between the parameters $\alpha$ and $\beta$ of his Bayesian posterior and the weighted average of his neighbors' parameters.\footnote{One may alternatively interpret agents to share opinions (mean) and precisions (variance) with each other rather than sharing distribution parameters. Those are equivalent modeling strategies, and we only need to use the relationships $y = \frac{\alpha}{\alpha+\beta}$ and $\sigma^2 = \frac{\alpha\beta}{(\alpha+\beta)^2(\alpha+\beta +1)}$ to  fully determine $\alpha$ and $\beta$. Algebraic manipulation yields $\alpha = - \frac{y (\sigma^2 + y^2 - y)}{\sigma^2}$ and $\beta = \frac{(\sigma^2 + y^2 - y) (y - 1)}{\sigma^2}$.}

In mathematical terms, the update rule is as follows
\begin{align}
\alpha_{i,t+1} &= b \left[\alpha_{i,t} + s_{i,t+1}^{(1)}\right] + (1 - b) \sum_{j} \hat{g}_{ij}\alpha_{j,t} \label{eq:update1} \\
\beta_{i,t+1} &= b \left[\beta_{i,t} + s_{i,t+1}^{(0)}\right] + (1 - b) \sum_{j} \hat{g}_{ij}\beta_{j,t}, \label{eq:update2}
\end{align} where $b \in [0,1]$.

Notice that when $b=1$, agents fully rely on the signals and behave like standard Bayesian agents. As $b$ approaches zero, agents are more influenced by the network as more weight is given to their neighbors' opinions. Moreover, let $\alpha_t = \left(\alpha_{1,t}, \alpha_{2,t}, \dots, \alpha_{n,t} \right)^{\top}$ and $\beta_t = \left(\beta_{1,t}, \beta_{2,t}, \dots, \beta_{n,t} \right)^{\top}$ denote the column vectors of length $n$ of agents beliefs parameters at time $t$, $\mathbb{I}$ be an identity matrix of dimension $n$ and $B = \text{diag}(b, b,\dots, b)$ be the diagonal Bayesian (or self-reliance) matrix. We can rewrite Equation \eqref{eq:update1} as follows
\begin{align}
\alpha_{t+1} &= B(\alpha_{t} + s_{t+1}^{(1)}) + (\mathbb{I}-B)\hat{g}\alpha_t \nonumber \\
				 &= \left(B + (\mathbb{I}-B)\hat{g}\right)\alpha_{t} + Bs_{t+1}^{(1)} \nonumber\\
				 &= W \alpha_{t} + Bs_{t+1}^{(1)}, 
\label{eq:alphaup}
\end{align}
and equation \eqref{eq:update2} as follows
\begin{equation}
\beta_{t+1} = W \beta_{t} + Bs_{t+1}^{(0)},
\label{eq:betaup}
\end{equation} where, $W = B + (\mathbb{I}-B)\hat{g}$ is a homogeneous row-stochastic matrix. Notice that as the graph $G$ induced by the adjacency matrix $g$ is assumed to be strongly connected, the graph induced by $W$ is trivially strongly connected as well.

\section{Theoretical results}
\label{sec:theoresults}

\subsection{Single individual case}
Before illustrating the network effects over the opinions when agents are exposed to ambiguous signals, we first focus on explaining what happens in the case of a single individual. In this regard, the following result shows that only two types of opinions may emerge when an agent interprets ambiguity under confirmatory bias.

\begin{proposition}[characterization]
\label{prop:singleagent}
If a single individual $i$ randomizes interpretation of ambiguous signals according to Equation~\eqref{eq:randomization}, while disregarding  neighbors' opinions ($b=1$), then their opinion converges to either $y_{r} = (1-\mu)\theta + \mu \gamma_{i}$ or $y_{l} = (1-\mu)\theta + \mu (1-\gamma_{i})$ almost surely, regardless of initial belief $(\alpha_{i,0},\beta_{i,0})$.
\end{proposition}

Both left and right biased opinions can be considered tail events and, therefore, may emerge with some positive probability, in the spirit of any classic zero-one law. Moreover, both opinions are biased as any limiting opinion is a weighted average that places weight $\mu$ on the confirmatory bias parameter $\gamma_{i}$ and weight $(1-\mu)$ on the true state $\theta$. Thus, both the fraction of ambiguous signals $\mu$ and confirmatory bias $\gamma_i$ are key misinformation drivers.

Another way to view this is by rewriting the expressions $y_l$ and $y_r$ as the true underlying state plus its biases
\[
 y_l = \theta + \underbrace{\mu \big((1-\gamma_i) - \theta\big)}_{\text{bias}}, \text{ and }
 \]
 \[
 y_r = \theta + \underbrace{\mu \big(\gamma_i - \theta\big)}_{\text{bias}}.
\]

 Considering the vaccine efficacy example, say the underlying efficacy is $\theta = 0.8$ (i.e., people who got the vaccine were at 80\% lower risk of contracting the disease), the confirmatory bias is $\gamma_i = 0.6$ and the fraction of ambiguous signals is $0.3$. As per the $y_r$ expression, the bias is $0.3 \times (0.6 - 0.8) = -0.06$, meaning that individual $i$ believes the efficacy is $0.8 - 0.06 = 0.74$. As per the $y_l$ expression, the bias is $0.3 \times \big((1-0.6) - 0.8\big) = -0.12$, meaning that individual $i$ believes the efficacy is $0.8 - 0.12 = 0.68$. Hence, both opinions are biased and both $\mu$ and $\gamma_i$ are important misinformation drivers.  From this perspective, in the absence of ambiguity ($\mu=0$), there would be no bias and no misinformation. Conversely, if $\mu > 0$, then misinformation could be fully mitigated if $\gamma_i = \theta$ (i.e. agents randomizing interpretation between 0 and 1 according to the true proportion $\theta$) if $\theta > 0.5$ or if $\gamma_i = 1-\theta$ if $\theta < 0.5$.

Based on this example, a first result of interest stemming from Proposition~\ref{prop:singleagent} is that, for any individual with confirmatory tendency, one opinion type is less biased than the other depending on the state $\theta$. This is generalized as follows.

\begin{corollary}[asymmetric bias]
\label{cor:efficiency}
For any individual with confirmatory tendency and for any ambiguity level, $y_{r}$ ($y_{l}$) is less biased than $y_{l}$ ($y_{r}$) if $\theta > \frac{1}{2}$ $\left(\theta < \frac{1}{2}\right)$. Conversely, both $y_{r}$ and $y_{l}$ are equally biased when $\theta = \frac{1}{2}$.
\end{corollary}

Generally, these two opinions are not equally distant from $\theta$ (see Example~\ref{ex:asymmetry-efficiency}). This is because the bias of each one depends on the relative size of $\theta$ and $\gamma_{i}$. As we are restricting attention to the case in which the agent has confirmatory tendency (i.e., $\gamma_{i}\geq\frac{1}{2}$), it is the case that agents make less mistakes when they are in the correct side of the spectrum. Therefore, ambiguity has to be low enough to not mislead agents and the state has to be high (or low) enough to nudge agents' opinions to the correct side.

\begin{example}
\label{ex:asymmetry-efficiency}
Suppose that a biased individual ($\gamma_{i}=1$) faces $20\%$ of ambiguous signals ($\mu=0.20$) and consider three particular values for the underlying state $\theta$, say: low ($\theta_{L} = 0.1$), medium ($\theta_{M}=0.5$) and high ($\theta_{H}=0.9$).\footnote{The arbitrary values chosen for $\theta$ in the Example~\ref{ex:asymmetry-efficiency} only mean to illustrate the results of Proposition 1 and Corollary 1 for a broad range of $\theta$. For all purposes, $\theta \in \Theta = [0,1].$} In this case, Proposition~\ref{prop:singleagent} and Corollary \ref{cor:efficiency} show that under state
\[
\theta_{L},  \begin{cases}
    y_r = (1-0.2)\times0.1 + 0.2\times1 = 0.28  \text{ is formed with probability $0$, and}\\
     y_l = (1-0.2)\times0.1 + 0.2\times(1-1) = 0.08 \text{ is formed with probability $1$.}
  \end{cases}
\]  Under state
\[
\theta_{M},  \begin{cases}
    y_r = (1-0.2)\times0.5 + 0.2\times1 = 0.60  \text{ is formed with some probability $p\in(0,1)$, and}\\
     y_l = (1-0.2)\times0.5 + 0.2\times(1-1) = 0.40 \text{ is formed with probability $1-p$.}
  \end{cases}
\] Finally, under state
\[
\theta_{H},  \begin{cases}
    y_r = (1-0.2)\times0.9 + 0.2\times1 = 0.92  \text{ is formed with probability $1$, and}\\
     y_l = (1-0.2)\times0.9 + 0.2\times(1-1) = 0.72 \text{ is formed with probability $0$.}
  \end{cases}
\]
\end{example}

Although $y_l$ and $y_r$ are not equidistant from $\theta$, the distance between $y_l$ and $y_r$ does not depend on $\theta$ and is equal to $\mu(2\gamma_{i}-1)$. Hence, for any given $\mu$, the distance between opinions $y_l$ and $y_r$ is maximal when $\gamma_i=1$. In this case, the distance becomes $\mu$. In the Example~\ref{ex:asymmetry-efficiency}, note that as $\gamma_i=1$, opinions distance is always $\mu=0.2$, regardless of state $\theta$.

Moreover, depending on the pair ($\theta$, $\mu$), we can show which opinion will be reached. Hence, we just need to determine which combinations of $\theta$ and $\mu$ are sufficient to allow both types of opinion to fall in the same side of the 0-1 spectrum (i.e. both $y_l$ and $y_r$ above 0.5 or both $y_l$ and $y_r$ below 0.5) and which combinations lead opinions to diverge in location (i.e. $y_r \geq 0.5$ and $y_l < 0.5$). These conditions lead to different regions of space $\Theta \times  M = [0,1]^2$ (unit square): region $L$, characterized by both low state $\theta$ and low ambiguity level $\mu$; region $R$, characterized by both high state and low ambiguity;  whereas region $\mathcal{W}$ is the complement of the union of $L$ and $R$. In mathematical terms, those partitions are characterized by the following:
\begin{align*}
R &= \left\{ (\theta,\mu) | \frac{1}{2} < \theta \leq 1 \text{ and } 0 \leq \mu < \frac{\theta - 0.5}{\gamma_{i} + \theta - 1} \right\}, \\
L &= \left\{(\theta,\mu) | 0 \leq \theta < \frac{1}{2} \text{ and } 0\leq \mu <\frac{\theta-0.5}{\theta-\gamma_{i}}\right\}, \\
\mathcal{W} &= [0,1]^{2}\setminus\{R\cup L\}. 
\end{align*}

We can state the following result regarding the general conditions for the emergence of each opinion type. 

\begin{proposition}[opinion-type emergence]
\label{prop:t1ct2c_areas}
For any individual with confirmatory tendency, if $(\theta,\mu) \in R$, then the limiting opinion is the right-biased one with probability 1. If $(\theta,\mu) \in L$, then the limiting opinion is the left-biased one with probability 1. If $(\theta,\mu) \in \mathcal{W}$, then the limiting opinion is a random variable whose possible values are $y_l$ and $y_r$.
\end{proposition}

This result holds regardless of initial beliefs and observed sequence of signals. For three cases of confirmation bias, Figure~\ref{fig:regions_t1ct2c} depicts the idea of Proposition~\ref{prop:t1ct2c_areas}. In case 1, when the agent is roughly impartial, case 2 when the agent has an intermediary level of confirmatory bias, and case 3 when agent is biased.

\vspace{2mm}
\begin{figure}[H]
  \begin{subfigure}[b]{0.32\textwidth}
    \includegraphics[width=\textwidth]{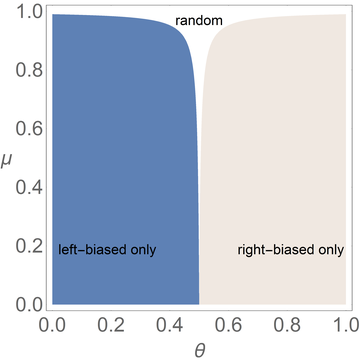}
    \caption{case 1: $\gamma_{i}= 0.505$}
    \label{fig:almostimpartial}
  \end{subfigure}
  \begin{subfigure}[b]{0.32\textwidth}
    \includegraphics[width=\textwidth]{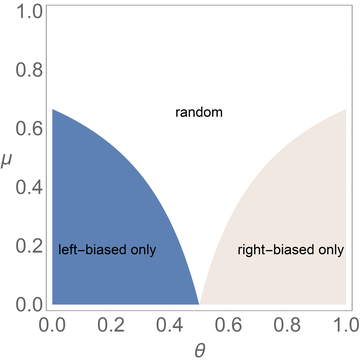}
    \caption{case 2: $\gamma_{i}= 0.750$}
    \label{fig:intermcase}
  \end{subfigure}
  \begin{subfigure}[b]{0.32\textwidth}
    \includegraphics[width=\textwidth]{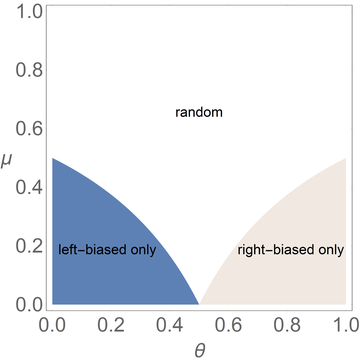}
    \caption{case 3: $\gamma_{i}= 1.000$}
    \label{fig:fullyconfirmatory}
  \end{subfigure}
  \caption{Parameter space and emergence of different types of consensus}
  \label{fig:regions_t1ct2c}
\end{figure}

Lightly shaded areas on the right represent the set of parameters $\mu$ (vertical axis) and $\theta$ (horizontal axis) ensuring the emergence of a right-biased opinion. Conversely, darkly shaded areas represent the set of parameters ensuring the emergence of left-biased opinion. In both areas, for a given level of confirmatory bias, the left (right)-biased opinion emerge with probability 1 if there is both low frequency of ambiguous signals and low (high) state (i.e., below (above) 0.5). On the other hand, the white area represents the combinations of $\theta$ and $\mu$ such that the opinion type becomes a random variable whose possible values are $y_l$ and $y_r$ (i.e., both opinion types may emerge with positive probability). Considering that when $(\theta,\mu) \in \mathcal{W}$ any of the two opinions may be formed, we define the probability $p $ with which an individual reaches the less-biased opinion following the results from Proposition~\ref{prop:singleagent} and Corollary~\ref{cor:efficiency}.

\begin{definition}
\label{def:problessbias}
For any given initial belief $(\alpha_{i,0},\beta_{i,0})$, ambiguity level $\mu \in (0,1)$ and confirmation bias $\gamma_{i} \geq \frac{1}{2}$, the probability of less-biased opinion forming is

$$
p = 
\begin{cases}
    P\displaystyle\left(\lim_{t \to \infty} y_{i,t} = y_l\right), &  \text{ when $\theta < 0.5$, or} \\
    P\displaystyle\left(\lim_{t \to \infty} y_{i,t} = y_r\right), & \text{ when $\theta > 0.5$.}
\end{cases}
$$
\end{definition}

Another interesting case stemming from Proposition~\ref{prop:singleagent} is the one in which bias cannot be overcome even when an agent is impartial.

\begin{corollary}[bias from impartiality]
\label{cor:impartial}
If an individual is impartial, then his limiting opinion is $(1-\mu) \, \theta + \mu \, \frac{1}{2}$ almost surely, regardless of his initial prior and the sequence of observed signals.
\end{corollary}

Impartiality does not overcome bias because it forces agents to set a disproportionate probability mass in the center of the spectrum $(0,1)$. Hence, impartiality makes agents excessively centrist instead of making them neutral toward possible states. This is a direct consequence of the Beta-Bernoulli conjugate family employed here that would not occur in a binary state space (i.e., $\Theta = \{0,1\}$).

Moreover, under impartiality, for any mass of ambiguity $\mu > 0$, if true state is located in the left side of the 0-1 spectrum ($\theta <\frac{1}{2}$), then opinion has a positive bias and lies in $\left(\theta, \frac{1}{2}\right)$. Conversely, if $\theta > \frac{1}{2}$, then opinion has a negative bias and lies in $\left(\frac{1}{2},\theta\right)$. The only instance when an individual learns the state is when $\theta=\frac{1}{2}$, a zero mass event if $\theta$ was drawn randomly from the interval $[0,1]$. The results presented so far both extend the intuition and contrast with Propositions 4 and 5 in \cite{Rabin1999} and with Propositions 2 and 3 in \cite{fhj2018}. This extends the intuition to the case in which the state is continuously distributed over the interval 0-1 and contrasts because impartiality can no longer help an individual overcome bias, as per the result above. 

Finally, at the other extreme, one could ask under what conditions an individual would reach an extreme opinion (i.e., either opinion 0 (extreme left) or opinion 1 (extreme right)). The next result shows that those cases can only be sustained under two extreme conditions: (i) the fraction of ambiguous signals is maximal ($\mu=1$) and, (ii) individual is biased ($\gamma_{i} =1$). 

\begin{corollary}[extreme opinions]
\label{cor:extremism}
For any $\theta \in \Theta$ and any initial belief ($\alpha_{i,0}, \beta_{i,0}$), any individual $i\in N$ will form an extreme opinion (either 0 or 1) if they are biased ($\gamma=1$) and the mass of ambiguity is maximal ($\mu=1$).
\end{corollary}

\subsection{Networked society}

Given the intuition of the single agent case described, one may ask what happens if agents also learned from their friends, besides learning from signals. This case imposes an extra challenge as the interpretation of ambiguous not only depends on the initial realization of signals but also on the influence of friends that potentially interpret ambiguity in different ways. ``\textit{Tug-of-war}'' played between left and right biases has one extra driver: the network externalities. 

Before discussing the implications of a network structure, we define the concept of consensus (Definition \ref{def:consensus}) and illustrate the social influence of agents, in terms of ergodicity of a Markov chain (Lemma \ref{lem:ergochain}), derived from the reliance weight matrix $W$. Appendix \ref{append:aux-defs-lemmas} contains the proof of the Lemma.

\begin{definition}[consensus]
\label{def:consensus}
Society reaches a consensus \textit{almost surely} for any initial beliefs if there is a $y$ such that, for every $\epsilon > 0$ and $i \in N$, 
\begin{equation*}
P\left(\lim\limits_{t\to\infty} |y_{i,t} - y| < \epsilon\right) = 1.
\end{equation*}
\end{definition}

\begin{lemma}[strong connectivity]
\label{lem:ergochain}
The $t$-th power of matrix $W$, $W^{t}$, converges to a unique row-stochastic matrix with unit rank (all rows the same) as $t$ tends to infinity, i.e.
\[
\lim_{t\to\infty} W^{t} = W^{\infty} = {\bf{1}}\pi^{\top} = \Pi,
\] where the invariant distribution $\pi$ is the normalized left eigenvector of the matrix $W$ associated to the unit eigenvalue, i.e. $\pi^{\top} W = \pi^{\top}$ and $\sum_{i}\pi_{i} = 1$.
\end{lemma}

A first case of interest is the limiting case in which individuals exclusively pay attention to friends. This represents the situation wherein agents disregard signals completely and are pure conformists. The consensus reached is slightly different from the classic DeGroot case, as the limiting opinion is not exactly a weighted average of the initial opinions, although is still very close to it. The discrepancy has to do with the fact that agents are exchanging opinions and precisions (parameters $\alpha$ and $\beta$ ). This is stated as follows.

\begin{proposition}[DeGroot consensus]
\label{prop:b-zero}
If the social network $G = (N, g)$ is strongly connected, and agents disregard all public signals ($b=0$), then society reaches consensus 
\[
\bar{y} = \frac{\sum_{j}\Pi_{ij}\alpha_{j,0}}{\sum_{j}\Pi_{ij}\left(\alpha_{j,0}+\beta_{j,0}\right)}, 
\]
for any $i \in N$ and where $\Pi$ is the invariant distribution matrix.
\end{proposition}

Section~\ref{sec:simulationanalysis} highlights the implications of this result wherein we explore the effects of priors heterogeneity on the probability of attaining the less-biased consensus. Next, we show that assuming strong connectivity, consensus is reached in this dynamic system and has a similar functional form of the individual limiting opinion in Proposition~\ref{prop:singleagent}.

\begin{proposition}[Network externality]
\label{prop:ntw-society}
With network externalities ($0<b<1$), sequences $\{y_{i,t}\}_{t=1}^{\infty}$ generated by the update rule converge almost surely to either right-biased consensus $\bar{y}_r = (1-\mu)\theta + \mu \bar{\gamma}$ or left-biased consensus $\bar{y}_{l} = (1-\mu)\theta + \mu (1-\bar{\gamma})$ for all $i \in N$ and where $\Pi$ is the invariant distribution matrix, and $\bar{\gamma}=\sum_{j}\Pi_{ij}\gamma_{j}$ for any $i \in N$.
\end{proposition}

Therefore, in a networked society, consensus is a \emph{weighted average} between the true state $\theta$ (with weight $1-\mu$) and the weighted average of confirmatory biases $\bar{\gamma}$ (with weight $\mu$). Society aggregates information efficientrly and no bias exists if $\mu=0$ or $\bar{\gamma} = \theta$ when $\theta>0.5$ or if $\bar{\gamma} = 1-\theta$ when $\theta<0.5$. Additionally, parameter $b$ impacts the vector of social influence through the invariant distribution $\pi$ of the matrix $W$ (see Lemma~\ref{lem:ergochain}) and therefore does impact consensus. Note that when $b=1$, social connection is lost, and polarization emerges with each individual opinion being a function of individual confirmation bias as in Proposition~\ref{prop:singleagent}. 

Moreover, the above results show that consensus type in this dynamic system is also a tail event (i.e., right-biased consensus will either almost surely emerges as the stable equilibrium or almost surely not emerge). If this does not emerge as an equilibrium of this system, then the left-biased consensus has truly emerged as the equilibrium. Figures ~\ref{fig:sampleopline} and ~\ref{fig:sampleopcircle} show the typical opinion sample paths (different simulations) of any agents' opinions in the line and wheel networks, respectively, and convergence to different consensus types (horizontal lines).

\begin{figure}[h]
        \centering
        \begin{subfigure}[b]{0.475\textwidth}
            \centering
            \input{graph-line-n3.tex}
            \vspace{10 mm}
            \caption{line network} 
            \label{fig:linenet}
        \end{subfigure}
        \hfill
        \begin{subfigure}[b]{0.475\textwidth}  
            \centering 
            \input{graph-circle-n3.tex}
            \caption{wheel network}    
            \label{fig:circlenet}
        \end{subfigure}
        \vskip\baselineskip
        \begin{subfigure}[b]{0.475\textwidth}   
            \centering 
            \includegraphics[width=.8\textwidth]{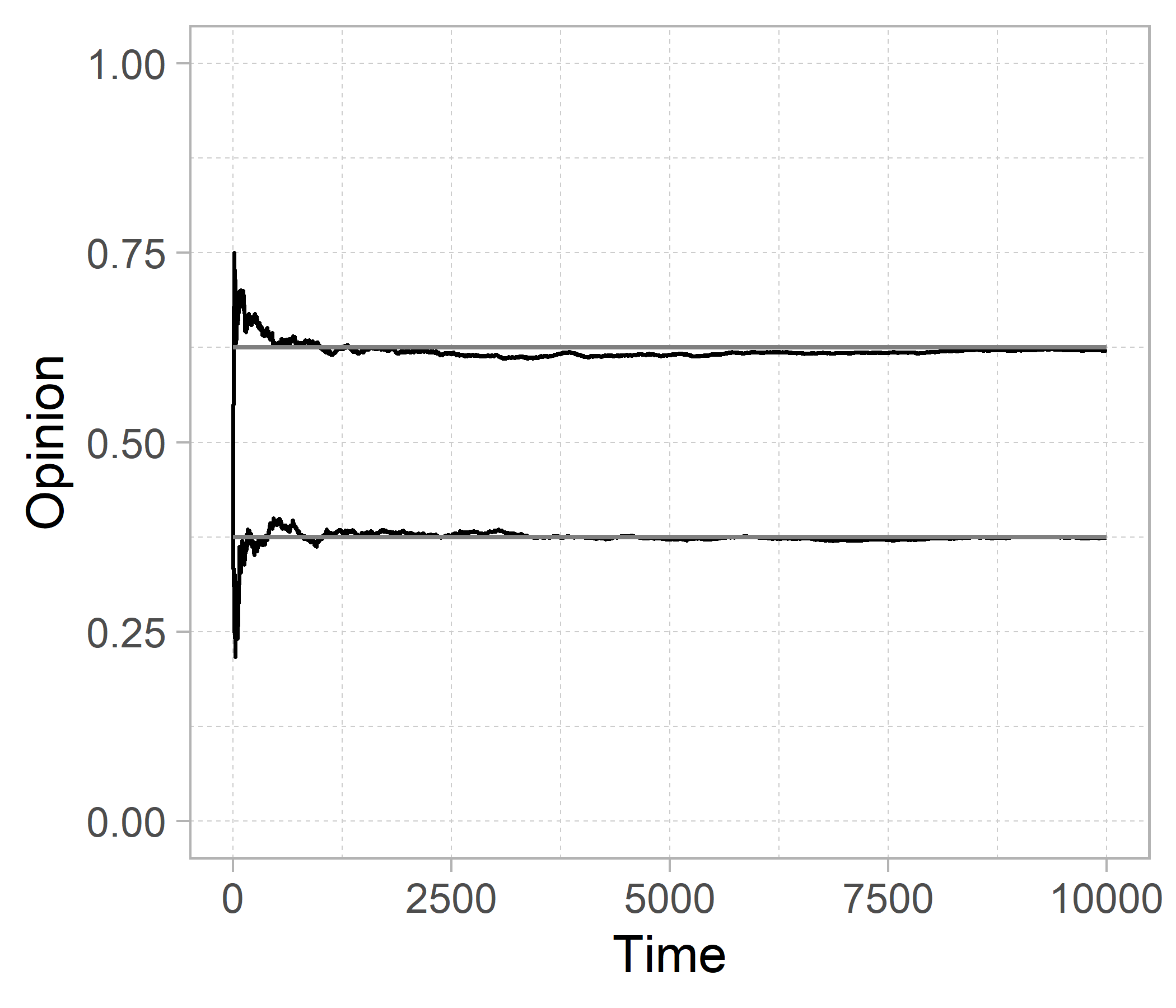}
            \caption{two simulated opinion paths (black) in a line network and theoretical consensuses (gray)}   
            \label{fig:sampleopline}
        \end{subfigure}
        \quad
        \begin{subfigure}[b]{0.475\textwidth}   
            \centering 
            \includegraphics[width=.8\textwidth]{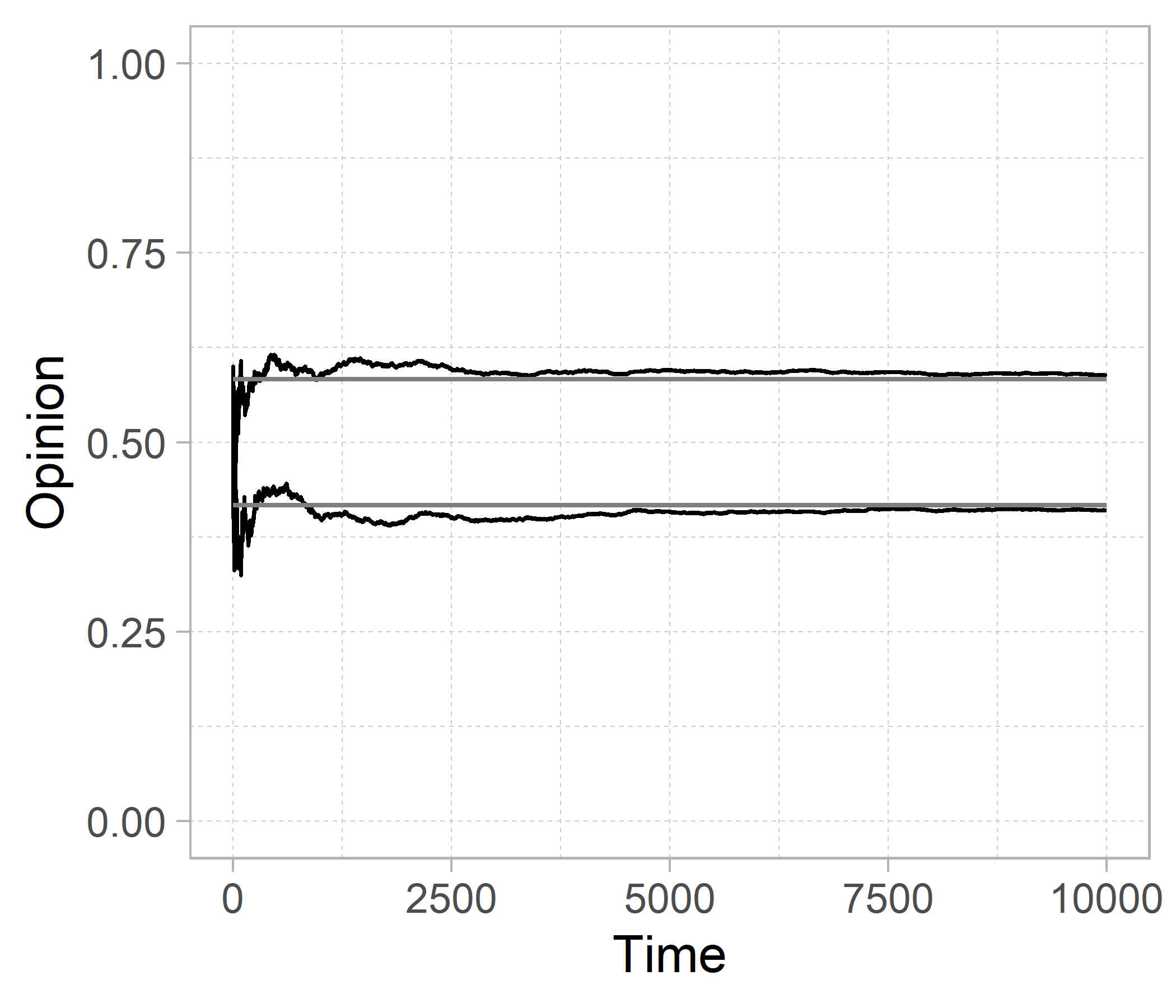}
            \caption{two simulated opinion paths (black) in a wheel network and theoretical consensuses (gray)}   
            \label{fig:sampleopcircle}
        \end{subfigure}
        \caption{four simulations with parameters $T=10,000$, $n=3$,  $\mu=\theta=b=0.5$, $\alpha_{i,0} = \beta_{i,0} = 1$ for all $i \in N$ (so $y_{i,0}= 0.5$ for any $i$) and $\gamma=\left(\gamma_{1},\gamma_{2},\gamma_{3}\right) = (0.8,1,0.2)$.} 
        \label{fig:t1ct2cinlinecircle}
\end{figure}


In terms of social efficiency, the next Section numerically characterizes the probability of emergence of efficient consensus. As this exercise is not trivial, we rely on simulations of the learning process described in Section~\eqref{sec:model} for selected classic network topologies and for different sets of parameters of interest and a Probit regression model to explore the variability of simulated data.

\section{Monte Carlo simulations: determining consensus type in classic networks}
\label{sec:simulationanalysis}

Determining the value of $p$ (Definition~\ref{def:problessbias}) analytically for a networked society is a challenging task due to several recursions in the opinion formation process. Initial prior distributions may be disproportionately skewed, with some agents being more partisan than others, leading to a propensity for interpreting ambiguous signals differently. In addition, the heterogeneity of priors can induce heterogeneity in the centrality of agents, with partisan agents potentially exerting disproportionate influence over others and amplifying interpreting conflicts in the network. Furthermore, agents may vary in the intensity of their confirmatory bias, with agents holding polar opposite biases being either directly connected or not. This variability can affect how much the heterogeneity of priors influences the interpreting conflict. Analytically computing $p$ thus becomes particularly complex when different partisan agents have different confirmatory biases. These challenges are highlighted in the examples, which demonstrate how signal interpretations depend not only on the stream of public signals observed by agents but also on other agents' beliefs and their location in the network.

\subsection{Initial beliefs}
\label{subsec:initprior}
This exercise reduces the dimension of initial beliefs into a single parameter $\tau \in \mathbb{R}_{\geq 0}$ that comprises both \emph{common} and \emph{heterogeneous} priors.
\label{subsec:priors}

\paragraph{\textbf{(a) Heterogeneous priors.}} Refer to the situation in which there are three types of agents at time $t=0$ with different initial prior distributions: centrists ($\mathcal{C}_{0}$), leftists ($\mathcal{L}_{0}$) and rightists ($\mathcal{R}_{0}$). To distinguish the agents, consider two parameters that intend to measure the degree of \textit{partisanship} of such agents $\tau_{l},\tau_{r} \in \mathbb{N_{+}}$. Hence, such groups are defined as follows: centrists, $\mathcal{C}_{0} = \left\{i \in N \, | \, \alpha_{i,0} = 1 \text{ and } \beta_{i,0} = 1 \right\}$, left-partisan, $\mathcal{L}_{0} = \left\{i \in N \, | \, \alpha_{i,0} = 1 \text{ and } \beta_{i,0} = 1+\tau_{l}\right\}$, and right-partisan, $\mathcal{R}_{0} = \left\{i \in N \, | \, \alpha_{i,0} = 1 + \tau_{r} \text{ and } \beta_{i,0} = 1\right\}$. Notice that the definition implies that initial opinions and precisions $y_{i,0} = \alpha_{i,0}(\alpha_{i,0}+\beta_{i,0})^{-1}$  and $\sigma_{i,0}^{-2} = \left(\alpha_{i,0}\beta_{i,0}\right)^{-1} \left(\alpha_{i,0}+\beta_{i,0}\right)^{2}\left(\alpha_{i,0}+\beta_{i,0}+1\right)$, respectively,
\[
y_{i,0} = \begin{cases}
    \dfrac{1}{2+\tau_{l}} , & \text{if } i \in \mathcal{L}_{0}\\[10pt]
     \dfrac{1}{2} , & \text{if } i \in \mathcal{C}_{0} \\[10pt]
      \dfrac{1+\tau_{r}}{2+\tau_{r}} , & \text{if } i \in \mathcal{R}_{0}
  \end{cases}
\]  

  and
  
\[
\sigma_{i,0}^{-2} = \begin{cases}
    \dfrac{6+5\tau_{l}+\tau_{l}^{2}}{1+\tau_{l}} , & \text{if } i \in \mathcal{L}_{0}\\[10pt]
     12 , & \text{if } i \in \mathcal{C}_{0}\\[10pt]
     \dfrac{6+5\tau_{r}+\tau_{r}^{2}}{1+\tau_{r}} , & \text{if } i \in \mathcal{R}_{0}.
  \end{cases}
\]  

\vspace{5mm}
Notice that $\lim_{\tau\to\infty} y_{i,0}$ is $0$, $\frac{1}{2}$ and $1$, whereas $\lim_{\tau\to\infty} \sigma_{i,0}^{-2}$ is $+\infty$, 12 and $+\infty$ for left-partisan, centrists and right-partisan, respectively.

\paragraph{\textbf{(b) Common prior.}} Refer to the situation in which priors parameters are identical across agents (i.e., $\alpha_{i,0} = \alpha\in\mathbb{R_{+}}$ and $\beta_{i,0} = \beta\in\mathbb{R_{+}}$ for all $i\in N$). Particularly, when $\alpha = \beta = 1$, all agents hold a uniform common prior over the unit interval. For any other value, say $\alpha = \beta = k > 1$, agents hold a symmetric bell-shaped common prior over the unit interval, centered at $0.5$. Moreover, as $k\to\infty$, the bell-shaped priors collapse to the point $0.5$ (i.e., the precision of the prior diverges) and all opinions are $y_{i,0} = 0.5$. These cases represent the situation wherein agents begin as \textit{centrists}, and the subsequent asymmetry of interpretation stems from the signals realizations.\footnote{To be more precise, interpretation neutrality does not exist as there is a non-neutral tie-break rule in Equation~\eqref{eq:randomization}. Thus, if the realization of the first public signal is $a$, then this signal will be interpreted as 1 by all agents, as per the tie-break rule. This is without loss of generality for the results presented in this work. The tie-break rule could have been defined in a way that the initial interpretation would be $0$ and intuition and conclusions would remain the same. Finally, if we established no tie-break rules, a more intricate update rule would be needed to maintain the prior when facing an ambiguous signal and opinions were exactly 0.5. Hence agents would keep opinions unchanged until some non-ambiguous realization occurs. In this case, the results would not also change as the neutral tie-break would promote some of the states, and the nature of the problem remains unchanged.} Conversely, when $\alpha_{i,0} = \alpha$ and $\beta_{i,0} = \beta$ for all $i\in N$ and $\alpha > \beta$ $\left(\beta > \alpha\right)$, the society holds a rightist (leftist) common prior (i.e., $y_{i,0} = y > 0.5$ ($y_{i,0} = y < 0.5$) and as $\frac{\alpha}{\beta}\to\infty$ ($\to 0$)), the bell-shaped priors collapse to the point $1$ ($0$) (i.e., the precision of the prior diverges and all opinions become extreme).

\subsection{Monte Carlo simulation.} To compute the empirical frequency of the emergence of the less-biased consensus ($\hat{p}$) for states and ambiguity level in $\mathcal{W}$, we simulate the learning process in Section~\eqref{sec:model} in selected classic networks ($G$) (Figure~\ref{fig:nets}). The number of simulations is described by $S \in \mathbb{N_{+}}$ and interaction time is $t\in\mathbb{N_{+}}$. 
\label{subsec:montecarlo}

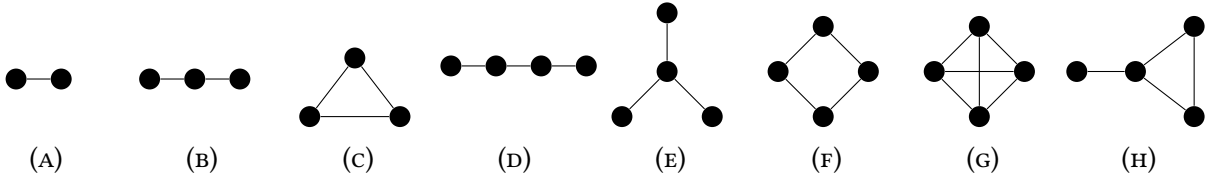
\begin{figure}[H]
  \begin{subfigure}[b]{0.12\textwidth}
  \centering
    \input{netA.tex}
    \vspace{5mm}
    \caption{}
    \label{fig:1}
  \end{subfigure}
  \begin{subfigure}[b]{0.12\textwidth}
  \centering
    \input{netB.tex}
        \vspace{5mm}
    \caption{}
    \label{fig:2}
  \end{subfigure}
  \begin{subfigure}[b]{0.12\textwidth}
  \centering
    \input{netC.tex}
    \caption{}
    \label{fig:3}
  \end{subfigure}
  \begin{subfigure}[b]{0.12\textwidth}
  \centering
    \input{netD.tex}
    \caption{}
    \label{fig:4}
  \end{subfigure}
  \begin{subfigure}[b]{0.12\textwidth}
  \centering
    \input{netE.tex}
    \caption{}
    \label{fig:5}
  \end{subfigure}
  \begin{subfigure}[b]{0.12\textwidth}
  \centering
   \input{netF.tex}
    \caption{}
    \label{fig:6}
  \end{subfigure}
  \begin{subfigure}[b]{0.12\textwidth}
  \centering
   \input{netG.tex}
    \caption{}
    \label{fig:7}
  \end{subfigure}
  \begin{subfigure}[b]{0.12\textwidth}
  \centering
   \input{netH.tex}
    \caption{}
    \label{fig:8}
  \end{subfigure}
  \caption{Classic networks -- $G$}
  \label{fig:nets}
\end{figure}

Each simulation allows some parameters to vary (see details below), so changes in $\hat{p}$ can be captured owing to changes in such parameters. However, in each simulation, parameter choices are identical for all networks. Hence, we can properly isolate effects on $\hat{p}$ owing to parameter variability. For any given simulation $S$, the realization of the public signals is also identical across networks. Therefore, given the choice of parameters in each simulation $S$, the simulated frequency of the less-biased consensus in a network $G$ follows Definition~\ref{def:problessbias}, and is computed as follows:
\begin{align}
\hat{p}_{G} = \frac{1}{S} \sum_{S} \mathbbm{1}\bigg\{&\left|\lim_{t\to\infty} y_{i,t}^{S,G} - \left((1-\mu_{S})\theta_{S} + \mu_{S}\gamma_{S}\right)\right| <\epsilon \text{ and } \theta_{S} > 0.5 \text{ , or } \nonumber\\
&\left|\lim_{t\to\infty} y_{i,t}^{S,G} - \left( (1-\mu_{S})\theta_{S} + \mu_{S}(1-\gamma_{S})\right)\right| <\epsilon \text{ and } \theta_{S} < 0.5\bigg\},
\end{align} for a small $\epsilon>0$.

\paragraph{\textbf{Partisanship effect.}} This exercise mainly aims to understand how degree of partisanship ($\tau$) affects the probability of less-biased consensus emerging. Hence, 
 both partisans are placed in the available nodes uniformly at random in each simulation. Thus, despite being in equal number, their level of centrality may differ in each simulation. Regarding the degree of partisanship, when $\tau_{l}=\tau_{r}=\tau = 0$, agents have a common uninformative prior (uniform distribution over the unit interval), and no partisan agents are found. Conversely, when $\tau_{l}=\tau_{r}= \tau > 0$, we have heterogeneous priors in which the degree of partisanship of both partisans is equally balanced. 

Moreover, to avoid an extra layer of heterogeneity, we allow all agents to be biased (i.e., $\gamma_i = \gamma_{S} = \gamma = 1$ for all $i$ and $S$). The benefit of fixing $\gamma=1$ across agents and simulations is that we know how each node is interpreting ambiguous signals and, hence, this allows us to study the effect of partisans centrality. Thus, this simulation assumes the following configuration:

\begin{itemize}
\setlength\itemsep{0.6em}

\item \textbf{Fixed parameters} (for all simulations) are as follows:

\begin{itemize}
\setlength\itemsep{0.2em}
\item \textbf{Learning}: $\gamma_{i} = \gamma = 1$ for all $i$,

\item \textbf{Bayesian}: $b = 0.5$,

\item \textbf{Duration}: $t = 700$,

\item \textbf{Priors}: $(\alpha_{i,0},\beta_{i,0})$ = $(1,1)$ for all $i \in N$, $|\mathcal{L}_{0}|=|\mathcal{R}_{0}|=1$,

\item \textbf{Information}: $\mu = 0.6$.

\end{itemize}

\item \textbf{Variable parameters} (for each simulation $S$) are

\begin{itemize}
\setlength\itemsep{0.2em}
\item \textbf{Information}: $\theta_{S} \in \{0.2, 0.8\} $ and 

\item \textbf{Initial prior}: $\tau_{l} = \tau_{r} = \tau_{S}$ such that $\tau_{S} \in \{0,1,10,30\}$,

\item \textbf{Partisans location}: partisans are placed uniformly at random in the nodes of each network.

\end{itemize}
\end{itemize}

Table~\ref{tab:sumstatsim} presents the summary statistics of the simulated data. Besides reporting the simulated $\hat{p}$ for each network, the table shows other variables that present variability across simulations $S$, as follows: 
\begin{enumerate}
    \item \emph{degree centrality} of both partisans in networks whose degree centrality variance is positive (i.e., networks (B), (D), (E) and (H)),
    \item \emph{open-mindedness} dummy variable valued at 1 when partisans with opposite beliefs are connected, and 0 if otherwise\footnote{ For any given network induced by some adjacency matrix $g$, a partisan agent $i\in N$ is considered open-minded if for some other partisan agent $j\in N$ with opposite belief, we have that $j\in N_{i}^{out}(g)$. Conversely, $i$ is narrow-minded if $j\notin N_{i}^{out}(g)$. Open-mindedness is defined quite similarly to \textit{heterophily} already established in social and economic networks literature. Both reflect the tendency of different people to connect with each other.},
    \item a dummy variable named \emph{first impression} (in reference to the work of \cite{Rabin1999}) that takes on value 1 when the realization of the very first public signal nudges the society toward the true state $\theta$ in that particular simulation and network\footnote{In mathematical terms, the first impression in simulation $S$ and network $G$ is defined as 
\[\text{FI}^{S,G} = \mathbbm{1}\{s_{1}^{S,G}=a \text{ or } s_{1}^{S,G}=1 \text{, if } \theta_{S}\geq0.5\} + \mathbbm{1}\{s_{1}^{S,G}=0  \text{, if } \theta_{S}<0.5\}.
\]},
    \item a dummy variable for when $\theta_S = 0.8$,
    \item and dummy variables for all levels of partisanship $\tau$.
\end{enumerate}

Besides analyzing the statistical relation established between $\hat{p}$ and different levels of $\tau$, Table~\ref{tab:probit-bench} presents and discusses the results of a Probit regression model wherein the dependent variable is a dummy variable that takes on the value 1 when less-biased consensus is formed in simulation $S$ after $t$ periods of agents interaction and the independent variables are (i) a dummy variable called \emph{partisan centrality advantage} (PCA) that takes on value 1 if the rightist (leftist) is more central than the leftist (rightist) when $\theta\geq0.5$ ($\theta<0.5$), (ii) open-mindedness dummy (OM), (iii) the first impression dummy, (iv) the dummy variable for when $\theta_S = 0.8$, and (v) all dummy variables for all different levels of partisanship $\tau$. The exercise aims at exploring the variability of these variables to assess their relative importance in terms of increasing the odds of the less biased consensus to be reached.   

\input{stats_0208.tex}

\input{probit_bench.tex}

Moreover, as the main goal is to understand the effect of partisanship on $\hat{p}$, Table~\ref{tab:simulation} shows $\hat{p}$ for different levels of $\tau$.

\begin{table}[ht]
\footnotesize
\centering
\begin{tabular}{>{\centering\arraybackslash}m{1cm} >{\centering\arraybackslash}m{2.3cm} >{\centering\arraybackslash}m{3cm} >{\centering\arraybackslash}m{1.2cm} >{\centering\arraybackslash}m{1.2cm}
>{\centering\arraybackslash}m{1.2cm}
>{\centering\arraybackslash}m{1.2cm}
>{\centering\arraybackslash}m{1.2cm}} \toprule
    \textbf{Size} & \textbf{Network Topology} & \textbf{Type} & \textbf{Label} & $\underset{(\tau=0)}{\hat{p}}$ & $\underset{(\tau=1)}{\hat{p}}$ & $\underset{(\tau=10)}{\hat{p}}$ & $\underset{(\tau=30)}{\hat{p}}$ \\ \midrule
     \multirow{ 1}{*}{$n=1$} & \input{netSA} &  single agent & (SA) & $0.702$ & - & - & -  \\ \midrule
     \multirow{ 1}{*}{$n=2$} & \input{netA} &  line (complete) & (A) & $0.702$ & $0.844$ & $0.844$ & $0.844$ \\ \midrule
     \multirow{ 2}{*}{$n=3$} & \input{netB}  & line & (B)  &$0.702$ & $0.802$ & $0.603$ & $0.610$  \\ 
     & \input{netC} &  wheel (complete) & (C)  &$0.702$ & $0.834$ & $0.852$ & $0.852$\\ \midrule
     & \input{netD} & line & (D)  &$0.702$ & $0.841$ & $0.739$ & $0.659$ \\
     & \input{netE} & star & (E)   &$0.721$ & $0.787$ & $0.643$ & $0.641$\\
     $n=4$ & \input{netF} &  wheel & (F)  &$0.702$ & $0.818$ & $0.884$ & $0.899$ \\
     & \input{netG} & complete\ & (G)  &$0.727$ & $0.801$ & $0.810$ & $0.812$ \\
     & \input{netH} & paw & (H)  &$0.720$ & $0.814$ & $0.706$ & $0.558$ \\ \midrule
     $S$&  &  &   & 21,040 & 21,040 & 21,040 & 21,040 \\ \bottomrule 
\end{tabular}
\caption{Simulated frequency of the emergence of less biased consensus $\hat{p}$.}
\label{tab:simulation}
\end{table}

Based on statistical and regression analyses of simulation data, we present results that hold for classic network structures presented above and for the case one draws a single pair in $\mathcal{W}$ uniformly at random (i.e., under no knowledge of $\theta$ or $\mu$). Although simulation results are not generalized to a broader range of network topologies, the structures analyzed are sufficiently general and have similar characteristics of real-world networks. Hence, as per the data from simulations with common prior ($\tau=0$, no partisanship), we can see that network structure has limited effect over $\hat{p}$. This evidence is stated as the following result.

\begin{result}[topology neutrality]
If society is biased ($\gamma = 1$) and have common prior (i.e., $\tau=0$), then network topology has no significant impact on $\hat{p}$.
\end{result}

The intuition of this result relies on the fact that as signals are public and all agents share the same bias intensity $\gamma_{i} = 1$, no interpretation diversity exists regardless of signals realization. If agents begin observing signal 1, then all agents will become more rightists and network externalities cannot countervail this effect anyhow. The same argument applies to all other signals, including the ambiguous one. Hence, this is identical to the case of a single individual learning from signals. Moreover, based on the data from simulations with a common prior ($\tau=0$, no partisanship) and low priors heterogeneity ($\tau=1$, low partisanship), partisanship seems to have a non-negative effect on consensus efficiency.

\begin{result}[low partisanship effect]
\label{res:lowpart}
In expected terms, a biased society with low degree of partisanship ($\tau=1$) can reach the less-biased consensus as the same biased society with no partisanship at all ($\tau=0$).
\end{result}

This can be seen in two ways: (i) there is a statistically significant difference between proportions in Table~\ref{tab:simulation} under $\tau =0$ and $\tau =1$, and (ii) coefficients of the dummy variable $\mathbbm{1}\{\tau=1\}$ are all positive and significant.\footnote{Note that the coefficients are with respect to the omitted dummy variable $\mathbbm{1}\{\tau=0\}$. The omission is needed so there is no perfect colinearization.}

Partisanship acts to counter the effect of initial ambiguous signals. Under no partisan influence, agents' interpretations depend exclusively on the signals. The realization of the initial signal is crucial to determine what bias opinions will have and, hence, it is determinant to consensus efficiency. Conversely, when some partisan agents are present, priors parameters $\alpha$'s and $\beta$'s are shifted up, by right- and left-partisans, respectively, which makes opinions more robust to initial signal realization. However, some "optimal" level of partisanship exists as high partisanship, for most topologies, has a nonmonotonic effect over the probability of emergence of the less-biased consensus. This result is generalized as follows.

\begin{result}[high partisanship effect]
\label{res:highpart}
In expected terms, a biased society with low partisanship ($\tau=1$) can reach the efficient consensus as the same biased society with high partisanship ($\tau=30$). Exceptions include the wheel and complete networks (i.e. (C), (F) and (G)) in which $\hat{p}$ is non-decreasing with partisanship.
\end{result}

This can be seen in two ways: (i) there is a statistically significant difference between the proportions in Table~\ref{tab:simulation} under $\tau = 1$ and $\tau = 30$, and (ii) the coefficients of the dummy variable $\mathbbm{1}\{\tau=30\}$ are higher than the coefficients of the dummy variable $\mathbbm{1}\{\tau=1\}$ for the referred networks.

Moreover, if an imbalance in partisanship exists, then partisan agents can unbalance opinions in the same way realization of the first signals do. More explicitly, a partisan agent with high degree of partisanship will almost never interpret ambiguous evidence in a way that disagrees with his beliefs and a similar effect applies to his neighbors. However, partisan agents might be more or less connected and even connected to each other. 

Naturally, in networks (A), (C), and (G), partisan agents are invariably open-minded as those networks are complete (i.e., all agents are connected with every other agent in the network, regardless of their types). Hence, analyzing the effect of OM in networks (B), (D), (E), (F), and (H) as those agents are not always connected. Table~\ref{tab:open-narrow} reports simulated probability $\hat{p}$ in those cases, and the next result is stated immediately.

 \input{open-narrow.tex} 

\begin{result}[open-minded partisans]
\label{res:open-narrow}
In expected terms, for biased agents, open-mindedness of partisans increases odds of less-biased consensus formation in networks (F) and (H). Conversely, narrow-mindedness of partisans increases the odds of the less-biased consensus formation in networks (B) and (E).
\end{result}

This can be seen in two ways: (i) a statistically significant (positive) difference between the proportions in Table~\ref{tab:open-narrow} under the open-minded (OM) and narrow-minded (NM) cases. That is, OM increases the odds relative to the NM case for networks (F) and (H), whereas NM increases the odds with respect to the OM case in networks (B) and (E). (ii) Coefficients of the dummy variable OM in the Probit regression are only positive for networks (F) and (H) and negative for the (B) and (E) networks.

The intuition relies on Results \ref{res:lowpart} and \ref{res:highpart}. In networks (B) and (E), OM would imply that one partisan is disproportionately more influential than the other (i.e., there would be a partisan centrality imbalance). In expected terms, the benefit of having a central partisan that induces the underlying true state is offset by the costs of having the opposite situation of the polar opposite partisan inducing misinformation. In network (F), OM means that partisans will moderate quickly as they are connected directly to each other but can still shield society from initial misleading signals, as discussed. Finally, in network (H), OM avoids centrality imbalance.

Another case of interest is the one in which agents are connected through a line. 

\begin{result}[line networks]
\label{res:linenets}
In expected terms, for biased agents connected through any sufficiently long line network ($n\geq3$), high partisanship ($\tau>1$) reduces the odds of reaching the less biased consensus. Moreover, for any given level of partisanship $\tau>0$, a longer line (higher $n$) increases the odds of reaching the less biased consensus. 
\end{result}

This can be seen in two ways: (i) there is a statistically significant difference between proportions in Table~\ref{tab:simulation} in both networks (B) and (D) for the proportions when $\tau > 1$ compared to the ones under $\tau \leq 1$, and (ii) the coefficients of the dummy variable $\mathbbm{1}\{\tau=30\}$ are negative for the referred networks.

A final result is related to the higher odds of reaching the less biased consensus when nodes are equally central and partisanship is high.

\begin{result}[regular networks]
\label{res:regularnets}
In expected terms, for biased agents connected through any regular network, higher levels of partisanship increase the odds of reaching the less biased consensus. 
\end{result}

Networks (A), (C), (F), and (G) are all regular. Table~\ref{tab:simulation} shows that for any of these networks, $\hat{p}$ increases as $\tau$ increases. Moreover, coefficients of the dummy variables for the levels $\tau=1$ to $\tau=30$ are in increasing order. This is because no partisan agent outweighs the opposing partisan easily in terms of influence. Thus, partisanship moderates the interpreting dispute by keeping the society close to the center of the 0-1 spectrum long enough, so informative signals accumulate and nudge the society toward the less-biased consensus.  


\section{Conclusions}
\label{sec:conclusions}

Confirmation bias is one of the most notorious cognitive biases documented and, as it is a systematic deviation from rationality, have a significant influence in the process of belief formation. In this sense, as social networks appear as a primary tool for many people to get informed and debate their worldviews, one could expect confirmatory bias to have some influence on the opinion formation. However, to date, how such phenomenon influences opinions in a networked environment has not been understood. To explore this topic, we consider a social learning model in which a fraction of signals external to the social network is ambiguous and open to idiosyncratic interpretation. The interpretation of these signals is affected by people's confirmatory biases. Moreover, we also allow agents to be influenced by their friends and set their beliefs to be a linear combination of the (biased) Bayesian posterior and the (also biased) friends' posteriors. 

My model shows that biased agents connected through social networks can only reach two types of consensus and both are biased, one to the left and the other to the right. However, one consensus type is less-biased than the other depending on the state. Moreover, I demonstrate that long-run learning is not attained even if agents are impartial when interpreting ambiguous signals. Those results contradict \cite{Rabin1999} and \cite{fhj2018} in which long-run learning takes place with a positive probability, and impartiality helps learning the state. Furthermore, the network effect presented, together with signal realizations, reinforces the interpreting ``\textit{tug-of-war}'' as agents might have their own biases confirmed (or mitigated) by other agents. 

Finally, as deriving the probability of emergence of the less-biased consensus is challenging, we relied on Monte Carlo simulations to show its determinants. We show that the presence of partisan agents in societies who suffer from confirmatory bias have two main effects on the expected consensus efficiency: (i) it helps countervail the misinterpretation of initial signals when there degree of partisanship is low and for that it increases expected efficiency; and (ii) exacerbates misinterpretation of signals when the degree of partisanship is high, reducing expected consensus efficiency. Moreover, we also show that open-mindedness of partisan agents, i.e., when partisans agree to exchange opinions with partisans with polar opposite beliefs, might reduce expected consensus efficiency in some social topologies. 

These results suggest that policies designed to mitigate partisanship and confirmatory bias effects in social networks have to consider also the positive and negative network externalities induced by them in different settings.


\clearpage
\nocite*
\bibliographystyle{ecta}
\bibliography{CBSN_bib_v02}

\clearpage
\appendix

\section{Beta-Bernoulli model and likelihood function of interpreted signals}
\label{sec:betalikeli}

At any time $t$, the belief of agent $i$ is represented by the Beta probability distribution with parameters $\alpha_{i,t}$ and $\beta_{i,t}$
\begin{equation}
f_{i,t}\left(\theta\right) = 
\begin{cases}
\dfrac{\Gamma\left(\alpha_{i,t}+\beta_{i,t}\right)}{\Gamma\left(\alpha_{i,t}\right)\Gamma\left(\beta_{i,t}\right)} \, \theta^{\alpha_{i,t}-1} (1-\theta)^{\beta_{i,t}-1} & \text{, for } \, 0<\theta<1 \\
0 & \text{, otherwise,}
\end{cases}
\label{eq:betapdf-append}
\end{equation} where $\Gamma(\cdot) $ is a Gamma function and the ratio of Gamma functions in the expression above is a normalization constant that ensures that the total probability integrates to 1. In this sense,
\begin{equation*}
f_{i,t}\left(\theta\right) \propto \theta^{\alpha_{i,t}-1} (1-\theta)^{\beta_{i,t}-1}.
\end{equation*}

The idiosyncratic likelihood induced by the agent $i$'s interpretation of the public signal $s_{t+1}$ is
\begin{equation*}
\ell_{i}(s_{t+1}|\theta) = \theta^{s_{i,t+1}^{(1)}} (1-\theta)^{s_{i,t+1}^{(0)}}
\end{equation*} and, therefore, the standard Bayesian posterior is computed as
\begin{equation*}
f_{i,t+1}(\theta | s_{t+1}) = \dfrac{\ell_{i}(s_{t+1}|\theta) \, f_{i,t}\left(\theta\right)}{\displaystyle\int_{\Theta} \ell_{i}(s_{t+1}|\theta) \, f_{i,t}\left(\theta\right) d\theta}.
\end{equation*} 
Since the denominator of the expression above is just a normalizing constant, the posterior distribution is said to be proportional to the product of the prior distribution and the likelihood function as
\begin{align*}
f_{i,t+1}(\theta | s_{t+1}) &\propto \ell_{i}(s_{t+1}|\theta) \, f_{i,t}\left(\theta\right) \\
							&\propto \theta^{\alpha_{i,t}+s_{i,t+1}^{(1)}-1} \left(1-\theta\right)^{\beta_{i,t} + s_{i,t+1}^{(0)}-1}.
\end{align*}

Therefore, the posterior distribution is
\[
f_{i,t+1}\left(\theta\right) = 
\begin{cases}
\dfrac{\Gamma\left(\alpha_{i,t+1}+\beta_{i,t+1}\right)}{\Gamma\left(\alpha_{i,t+1}\right)\Gamma\left(\beta_{i,t+1}\right)} \, \theta^{\alpha_{i,t+1}-1} (1-\theta)^{\beta_{i,t+1}-1} & \text{, for } \, 0<\theta<1 \\
0 & \text{, otherwise}, 
\end{cases}
\] where $\alpha_{i,t+1} = \alpha_{i,t} + s_{i,t+1}^{(1)}$ and $\beta_{i,t+1} = \beta_{i,t} + s_{i,t+1}^{(0)}$.

\section{Beta Distribution: Mode, Mean, Median}
\label{sec:betadist}

\subsection*{Mode} The mode of a random variable beta-distributed is the value that appears most often. It is the value $\theta$ at which its probability density function takes its maximum value. As per Equation~\eqref{eq:betapdf-append}, the mode $\theta^{mod}_{i,t}$, for any $i$ at any point in time $t$, is the $\argmax_{\theta} f_{i,t}(\theta)$. Computed as

\begin{equation*}
\dfrac{df_{i,t}}{d\theta} = \dfrac{\Gamma\left(\alpha_{i,t}+\beta_{i,t}\right)}{\Gamma\left(\alpha_{i,t}\right)\Gamma\left(\beta_{i,t}\right)} \left[(\alpha_{i,t}-1) \theta^{\alpha_{i,t}-2} (1-\theta)^{\beta_{i,t}-1} -  \theta^{\alpha_{i,t}-1}(\beta_{i,t}-1) (1-\theta)^{\beta_{i,t}-2} \right] = 0. 
\end{equation*}

Implying that 
\begin{equation*}
(\alpha_{i,t}-1) \theta^{\alpha_{i,t}-2} (1-\theta)^{\beta_{i,t}-1} -  \theta^{\alpha_{i,t}-1}(\beta_{i,t}-1) (1-\theta)^{\beta_{i,t}-2} = 0,
\end{equation*} and therefore 
\begin{equation}
\theta^{mod}_{i,t} = 
\begin{cases}
\dfrac{\alpha_{i,t} -1}{\alpha_{i,t} + \beta_{i,t} - 2} & \text{, for } \, \alpha_{i,t},\beta_{i,t} > 1 \\
0 & \text{, for } \, \alpha_{i,t} = 1, \beta_{i,t} > 1 \\
1 & \text{, for } \, \alpha_{i,t} > 1, \beta_{i,t} = 1 \\
\text{any value in} \, (0,1) & \text{, for } \, \alpha_{i,t}, \beta_{i,t} = 1
\end{cases}
\end{equation}

\subsection*{Mean} The mean of a random variable Beta-distributed, denoted by $\theta_{i,t}^{mean}$ for any $i$ and $t$, is computed as follows
\begin{align}
\theta_{i,t}^{mean} &= \int_{0}^{1} \theta \dfrac{\Gamma\left(\alpha_{i,t}+\beta_{i,t}\right)}{\Gamma\left(\alpha_{i,t}\right)\Gamma\left(\beta_{i,t}\right)} \theta^{\alpha_{i,t}-1} (1-\theta)^{\beta_{i,t}-1} d\theta \nonumber\\
 &= \dfrac{\Gamma\left(\alpha_{i,t}+\beta_{i,t}\right)}{\Gamma\left(\alpha_{i,t}\right)\Gamma\left(\beta_{i,t}\right)} \int_{0}^{1} \theta^{(\alpha_{i,t}+1)-1} (1-\theta)^{\beta_{i,t}-1} d\theta \nonumber\\
 &= \dfrac{\Gamma\left(\alpha_{i,t}+\beta_{i,t}\right)}{\Gamma\left(\alpha_{i,t}\right)\Gamma\left(\beta_{i,t}\right)}  \dfrac{\Gamma\left(\alpha_{i,t}+1\right)\Gamma\left(\beta_{i,t}\right)}{\Gamma\left(\alpha_{i,t}+\beta_{i,t}+1\right)}\nonumber\\ 
 &= \dfrac{\Gamma\left(\alpha_{i,t}+\beta_{i,t}\right)}{\Gamma\left(\alpha_{i,t}\right)\Gamma\left(\beta_{i,t}\right)}  \dfrac{\alpha_{i,t}\Gamma\left(\alpha_{i,t}\right)\Gamma\left(\beta_{i,t}\right)}{(\alpha_{i,t}+\beta_{i,t})\Gamma\left(\alpha_{i,t}+\beta_{i,t}\right)}= \dfrac{\alpha_{i,t}}{\alpha_{i,t}+\beta_{i,t}}.
\end{align}

\subsection*{Median} There is no general closed formula for the median of the beta
distribution for arbitrary values of the parameter $\alpha_{i,t}$ and $\beta_{i,t}$. The median, denoted by $\theta^{med}_{i,t}$, is the function that satisfies
\begin{equation*}
\dfrac{\Gamma\left(\alpha_{i,t}+\beta_{i,t}\right)}{\Gamma\left(\alpha_{i,t}\right)\Gamma\left(\beta_{i,t}\right)} \int_{0}^{\theta_{i,t}^{med}} \theta^{\alpha_{i,t}-1} (1-\theta)^{\beta_{i,t}-1} = \frac{1}{2}.
\end{equation*}

An accurate approximation of the value of the median of the beta distribution, for both $\alpha_{i,t}, \beta_{i,t} \geq 1$, is given by 
\begin{equation}
\theta_{i,t}^{med} = \frac{\alpha_{i,t} - \frac{1}{3}}{\alpha_{i,t} + \beta_{i,t} - \frac{2}{3}}.\footnote{With relative error of less than 4\%, rapidly decreasing to zero as both shape parameters increase.}
\end{equation}

Therefore, if $1 < \alpha_{i,t} < \beta_{i,t}$, then $\theta_{i,t}^{mod} < \theta_{i,t}^{med} < \theta_{i,t}^{mean}$. If $1 < \beta_{i,t} < \alpha_{i,t}$, then the order of the inequalities is reversed. Finally, it is trivial to see that those three statistical measures are asymptotically equal as $\alpha_{i,t},\beta_{i,t} \to \infty$.

\pagebreak
\section{Auxiliary Definitions and Lemmas}
\label{append:aux-defs-lemmas}

\noindent \textbf{Proof of Lemma \ref{lem:ergochain}.} In order to see how $W^{t}$ behaves as $t$ grows large, I rewrite $W$ using its diagonal decomposition. In particular, let $v$ be the squared matrix of left-hand eigenvectors of $W$ and $D = (d_1,d_2, \dots, d_n)^{\top}$ the eigenvector of size $n$ associated to the unity eigenvalue $\lambda_1=1$. Without loss of generality, we assume the following normalization $\pmb{1}^{\top} D = 1$. Therefore,
$W = v^{-1}\Lambda v$, where $\Lambda = \text{diag}(\lambda_1, \lambda_2,\dots,\lambda_n)$ is the squared matrix with eigenvalues on its diagonal, ranked in terms of absolute values, i.e. $|\lambda_{1}| \geq |\lambda_{2}| \geq \dots \geq |\lambda_{n}|$. More generally, for any time $t$ we write
\begin{equation*}
W^{t} = v^{-1}\Lambda^{t} v.
\end{equation*}

Since $v^{-1}$ has ones in all entries of its first column, it follows that
\[
W^{t}_{ij} = d_{j} + \displaystyle \sum_{r}\lambda_{r}^{t}v_{ir}^{-1}v_{rj},
\]
for each $r$, where $\lambda_r$ is the $r$-th largest eigenvalue of $W$. Therefore, $\lim_{t \to \infty} W^{t}_{ij} = D\pmb{1}^{\top}$, i.e. each row of $W^{t}$ for all $t \geq \bar{t}$ converge to $D$, which coincides with the stationary distribution. Moreover, if the eigenvalues are ordered the way we have assumed, then $\lVert W^{t} - D\pmb{1}^{\top}\rVert = o(|\lambda_2|^{t})$, i.e. the convergence rate will be dictated by the second largest eigenvalue, as the others converge  to zero more quickly as $t$ grows. \hfill $\blacksquare$

\begin{lemma}
\label{lem:opinionint}
The opinion of every agent $i$ in any point in time $t$, $y_{i,t}$, can be written as
\[
y_{i,t} = \dfrac{\sum_{j=1}^{n}W_{ij}^{t}\alpha_{j,0} + b K(i,t)}{\sum_{j=1}^{n}W_{ij}^{t}\left(\alpha_{j,0}+\beta_{j,0}\right) + b L(i,t)},
\] where $K(i,t) = \displaystyle \sum_{k=0}^{t-1}\sum_{j=1}^{n}W_{ij}^{k}s_{j,t-k}^{(1)}$ and $L(i,t) = \displaystyle \sum_{k=0}^{t-1} \sum_{j=1}^{n} W_{ij}^{k} \left(s_{j,t-k}^{(0)} + s_{j,t-k}^{(1)}\right)$.
\end{lemma}

\begin{proof}
The update process of both parameters described by the equations~\eqref{eq:alphaup} and~\eqref{eq:betaup} can be solved iteratively for any period $t$ as
\begin{align}
\alpha_{t} &= W^{t}\alpha_{0} + \sum_{k=0}^{t-1}W^{k} B s_{t-k}^{(1)} \label{eq:alphaupfwd} \\
\beta_{t} &= W^{t}\beta_{0} + \sum_{k=0}^{t-1}W^{k} B s_{t-k}^{(0)} \label{eq:betaupfwd}.
\end{align}

In agebraic formulation, we have that each entry of the vector in equation ~\eqref{eq:alphaupfwd} can be written as
\begin{align}
\alpha_{i,t} &= \sum_{j=1}^{n}W_{ij}^{t}\alpha_{j,0} + \sum_{k=0}^{t-1}\sum_{j=1}^{n} W_{ij}^{k} b s_{j,t-k}^{(1)} \nonumber\\
				&=  \sum_{j=1}^{n}W_{ij}^{t}\alpha_{j,0} + b \sum_{k=0}^{t-1}\sum_{j=1}^{n} W_{ij}^{k} s_{j,t-k}^{(1)} \nonumber\\
				&= \sum_{j=1}^{n}W_{ij}^{t}\alpha_{j,0} + b \, \sum_{k=0}^{t-1}\sum_{j=1}^{n} W_{ij}^{k} s_{j,t-k}^{(1)} \nonumber\\
				&= \sum_{j=1}^{n}W_{ij}^{t}\alpha_{j,0} + b \, K(i,t). \label{eq:alphaupalgfwd}
\end{align}

Similarly for the expression $\alpha_{i,t} + \beta_{i,t}$ using both equations \eqref{eq:alphaupfwd} and \eqref{eq:betaupfwd} as follows
\begin{align}
\alpha_{i,t} + \beta_{i,t} & = \sum_{j=1}^{n}W_{ij}^{t} \left(\alpha_{j,0} + \beta_{j,0}\right) + b \, \sum_{k=0}^{t-1}\sum_{j=1}^{n} W_{ij}^{k} \left(s_{j,t-k}^{(0)}+s_{j,t-k}^{(1)}\right) \nonumber \\
						&= \sum_{j=1}^{n}W_{ij}^{t} \left(\alpha_{j,0} + \beta_{j,0}\right) + b \, L(i,t). \label{eq:alphaplusbetaalgfwd}
\end{align}

Therefore, from the definition of opinion we have that $y_{i,t} = \frac{\alpha_{i,t}}{\alpha_{i,t}+\beta_{i,t}}$ and the statement is proven.
\end{proof}

\begin{lemma}
\label{lem:expect-indfun}
Let $k\in[0,1]$, $X_{1}, X_{2}, \dots, X_{t}$ be a convergent sequence of i.n.i.d. random variables such that $\mathbb{P}(X_{t} \geq x) = p_t$ with $\lim_{t \to \infty} p_t = p$, and $u_{1}, u_{2}, \dots, u_{t}$ be i.i.d.  $U[0,1]$ random variables with cumulative density function $F_u$. Moreover, assume the pair $(X_{t}$, $u_{t})$ is independent for any $t$. In this case, the expressions \; $\mathbb{E}\left[\mathbbm{1}\{u_{t} \leq \mathbbm{1}\{X_{t} \geq x\}k\}\right]$ \; and $\mathbb{E}\left[\mathbbm{1}\{u_{t} \leq \mathbb{E}\left[\mathbbm{1}\{X_{t} \geq x\}\right]k\}\right]$ are equal.
\end{lemma}
\begin{proof}
The first expression can be written as
\[
\mathbb{E}\left[\mathbbm{1}\{u_{t} \leq \mathbbm{1}\{X_{t} \geq x\}k\}\right]=(1-p)\mathbb{E}\left[\mathbbm{1}\{u_t\le0\}\right]+p\mathbb{E}\left[\mathbbm{1}\{u_t\le k\}\right]=p F_{u}(k) = pk.
\]
The second expression simplifies to
\[
\mathbb{E}\left[\mathbbm{1}\{u_{t} \leq \mathbb{E}\left[\mathbbm{1}\{X_{t} \geq x\}\right]k\}\right]=\mathbb{E}\left[\mathbbm{1}\{u_{t} \leq (1-p)0+pk\}\right]=\mathbb{E}\left[\mathbbm{1}\{u_{t} \leq pk\}\right]=pk.
\]
\end{proof}

\begin{lemma}[Convergence]
\label{lem:convergence}
The sequences $\{\{y_{i,t}\}_{i=1}^{n}\}_{t=1}^{\infty}$ generated by the update rule converge almost surely as $t \to \infty$.
\end{lemma}
\begin{proof}

For the individual case, \cite{siegrist2021} (Section 12.8.5) shows that there is an equivalence between the Beta-Bernoulli process and the Pólya's urn process. In the Pólya's urn proccess the sequence of random variables (drawn balls' colors) is not independent, but is exchangeable. Thus, the joint distribution of the interpreted signals (colors) is invariant under a permutation. Thus, the sequence of the proportion of signals interpreted as 1 is a martingale, and standard martingale convergence theorems ensure the convergence of this process. 

For the networked case, Lemmas (1) and (2) in \cite{jmst2012} prove convergence of this process for a general case based on the same assumption that the social interaction matrix $W$ is strongly connected and, for that, it always has at least one eigenvalue equal to 1 and that there exists a non-negative left eigenvector $v$ corresponding to this eigenvalue. As a result, they show that $\sum_{i=1}^{n}v_i f_{i,t}(\theta^{*})$ is a submartingale with respect to the filtration $\mathcal{F}_t$ (interpreted signals).

\end{proof}

\section{Proofs of main propositions and corollaries}
\label{append:statregs}

\noindent \textbf{Proof of Proposition~\ref{prop:singleagent}.}
\begin{align}
\lim_{t\to\infty} y_{i,t} & = \lim_{t\to\infty} \, \dfrac{\alpha_{i,0} + \sum_{k=1}^{t}s_{i,k}^{(1)}}{\alpha_{i,0} + \sum_{k=1}^{t}s_{i,k}^{(1)} + \beta_{i,0} + \sum_{k=1}^{t}s_{i,k}^{(0)}} \nonumber\\
 &=  \lim_{t\to\infty} \, \dfrac{\alpha_{i,0} + \sum_{k=1}^{t} \left(\mathbbm{1}\{s_{k}=1\} +  \mathbbm{1}\{s_{k}=a\}\mathbbm{1}\{u_{k} \leq \psi_{i,k}\}\right)}{\alpha_{i,0} + \beta_{i,0} + \sum_{k=1}^{t} \left(\mathbbm{1}\{s_{k}=1\} + \mathbbm{1}\{s_{k}=0\} +  \mathbbm{1}\{s_{k}=a\}\right)} \nonumber \\
 &= \lim_{t\to\infty} \, \dfrac{\frac{\alpha_{i,0}}{t} + \frac{1}{t} \sum_{k=1}^{t} \left(\mathbbm{1}\{s_{k}=1\} +  \mathbbm{1}\{s_{k}=a\}\mathbbm{1}\{u_{k} \leq \psi_{i,k}\}\right)}{\frac{\alpha_{i,0} + \beta_{i,0}}{t} + \frac{1}{t} \sum_{k=1}^{t} \left(\mathbbm{1}\{s_{k}=1\} + \mathbbm{1}\{s_{k}=0\} +  \mathbbm{1}\{s_{k}=a\}\right)} \nonumber \\
 &= \dfrac{\mathbb{E}_{t}\left[\mathbbm{1}\{s_{t}=1\}\right]  +  \mathbb{E}_{t}\left[\mathbbm{1}\{s_{t}=a\}\right] \lim_{t\to\infty} \frac{1}{t} \sum_{k=1}^{t}\left(\mathbbm{1}\{u_{k}\leq \psi_{i,k}\}\right)}{\mathbb{E}_{t}\left[\left(\mathbbm{1}\{s_{t}=1\}\right] + \mathbb{E}_{t}\left[\mathbbm{1}\{s_{t}=0\}\right] +  \mathbb{E}_{t}\left[\mathbbm{1}\{s_{t}=a\}\right]\right)} \nonumber \\
 &= (1-\mu) \theta + \mu \lim_{t\to\infty} \frac{1}{t} \sum_{k=1}^{t}\left(\mathbbm{1}\{u_{k}\leq \psi_{i,k}\}\right) \nonumber \\
 &= (1-\mu) \theta + \mu \lim_{t\to\infty} \frac{1}{t} \sum_{k=1}^{t}\left(\mathbbm{1}\{u_{k}\leq \gamma_{i} \, \mathbbm{1}\{y_{i,k-1}\geq0.5\} + (1-\gamma_{i}) \, \mathbbm{1}\{y_{i,k-1}<0.5\}\}\right) \nonumber \\
 &=(1-\mu) \theta + \mu \mathbb{E}_{t}\left[\mathbbm{1}\{u_{t}\leq \mathbb{E}_{t}\left[\gamma_{i} \, \mathbbm{1}\{y_{i,t-1}\geq0.5\} + (1-\gamma_{i}) \, \mathbbm{1}\{y_{i,t-1}<0.5\}\}\right]\right] \nonumber \\
 &= (1-\mu) \theta + 
 \mu \mathbb{E}_{t}\left[\mathbbm{1}\{u_{t}\leq\mathbb{E}_{t}\left[\mathbbm{1}\{y_{i,t-1}\geq0.5\}\right]\left(2\gamma_{i}-1\right) + 1 - \gamma_{i}\}\right] \nonumber
\end{align}

According to Lemma~\ref{lem:convergence}, convergence ensures that $\mathbb{E}_{t}\left[\mathbbm{1}\{y_{i,t-1}\geq0.5\}=\mathbb{P}\left(y_{i,\infty}\geq0.5\right)\right]$ either takes on value 1 or 0. For simplicity, say the first case is denoted by A, and the second by B. Moreover, Lemma 4 also implies that the sequence
\begin{equation}
\label{eq:indicu}
\left\{\mathbbm{1}\{u_{k}\leq \gamma_{i} \, \mathbbm{1}\{y_{i,k-1}\geq0.5\} + (1-\gamma_{i}) \, \mathbbm{1}\{y_{i,k-1}<0.5\}\}\right\}_{k=1}^{t}    
\end{equation}
is composed by the terms
\begin{equation}
\label{eq:indicu-geq0.5}
\left\{\mathbbm{1}\{u_{k}\leq \gamma_{i}\}\right\}_{k \in \bar{k}}    
\end{equation} and by the terms
\begin{equation}
\label{eq:indicu-l0.5}
\left\{\mathbbm{1}\{u_{k}\leq 1-\gamma_{i}\}\right\}_{k \in \underline{k}},   
\end{equation}
         
where $\bar{k} = \{t \in \mathbbm{N}_{+} \; | \; y_{i,t} \geq 0.5\}$ and $\underline{k} = \{t \in \mathbbm{N}_{+} \; | \; y_{i,t} < 0.5\}$, for any $i$.

Intuitively, this means that before converging, opinions fluctuate between 0 and 1, and this binary sequence takes on a value of 1 when the opinion is above or equal to 0.5 and 0 otherwise. Since we know that opinions will either converge to $y \geq 0.5$ or to $y < 0.5$, we can say that the sequence in equation \eqref{eq:indicu}, at some point, will turn into a sequence composed only of the terms in equations \eqref{eq:indicu-geq0.5} or \eqref{eq:indicu-l0.5}. Thus, we have an i.n.i.d. sequence of Bernoulli random variables $X_k \sim \text{Bern}(p_k)$ with $\lim_{k\to\infty} p_k = p = \gamma$ or $1 - \gamma$, and standard (weak or strong) law of large numbers can be used to show that
         \[
          \lim_{t\to\infty} \frac{1}{t} \sum_{k=1}^{t} X_k \to \mathbbm{E}_t\left(X_t\right) = p.
         \] 

Therefore, 
\begin{align}
 \lim_{t\to\infty} y_{i,t} &= \begin{cases} 
      (1-\mu) \theta + \mu \mathbb{E}_{t}\left[\mathbbm{1}\{u_{t}\leq \gamma_{i}\}\right] & , \text{if} \,\, A \nonumber \\
      (1-\mu) \theta + \mu \mathbb{E}_{t}\left[\mathbbm{1}\{u_{t}\leq 1 - \gamma_{i}\}\right] & , \text{if} \,\, B 
  \end{cases} \\
  &= \begin{cases} 
      (1-\mu) \theta + \mu F_{u}\left(\gamma_{i}\right) & , \text{if} \,\, A \nonumber \\
      (1-\mu) \theta + \mu F_{u}\left(1-\gamma_{i}\right)& , \text{if} \,\, B 
  \end{cases} \\
  	    &= \begin{cases} 
      (1-\mu) \theta + \mu \gamma_{i} & , \text{if} \,\, A  \label{eq:singleagentproof} \\
      (1-\mu) \theta + \mu \left(1-\gamma_{i}\right) & , \text{if} \,\, B 
  \end{cases}
\end{align} \hfill $\blacksquare$

\noindent \textbf{Proof of Proposition~\ref{prop:t1ct2c_areas}.} The claim is supported by the solution of two systems of inequalities $S_{1}$ (for right-biased opinion) and $S_{2}$ (for left-biased opinion) below.
\begin{equation*}
\begin{minipage}[c]{0.5\linewidth}
\[
S_{1} = \begin{cases}
    (1-\mu)\theta + \mu \gamma_{i} > \frac{1}{2} \\
    (1-\mu)\theta + \mu (1-\gamma_{i}) > \frac{1}{2} \\
    0 < \mu \leq 1 \\
    0 \leq \theta \leq 1 \\
    \frac{1}{2}< \gamma_{i} \leq 1
  \end{cases}
\]  
\end{minipage}
\begin{minipage}[c]{0.5\linewidth}
\[
S_{2} = \begin{cases}
    (1-\mu)\theta + \mu \gamma_{i} < \frac{1}{2} \\
    (1-\mu)\theta + \mu (1-\gamma_{i}) < \frac{1}{2} \\
    0 < \mu \leq 1 \\
    0 \leq \theta \leq 1 \\
    \frac{1}{2}< \gamma_{i} \leq 1
  \end{cases}
\]  
\end{minipage}
\end{equation*}

The solution of those systems, together with the equation~\eqref{eq:singleagentproof} in Proof of proposition~\ref{prop:singleagent} ensure the uniqueness of opinion types in the parameter spaces defined in the statement. \hfill $\blacksquare$

\vspace{5mm}
\noindent \textbf{Proof of Corollary~\ref{cor:efficiency}.} From Proposition \ref{prop:singleagent}, we can write both right-biased and left-biased opinions as $\theta + \mu (\gamma_{i}-\theta)$ and $\theta + \mu(1-\gamma_{i}-\theta)$, respectively, where the second term in each expression represents their respective biases. From those expressions, we can see that both sign and magnitude of those biases naturally depend on the relative size of $\theta$ and $\gamma_{i}$. For both biases to be positive, we need $\theta < \min\{\gamma_{i},1-\gamma_{i}\} = 1-\gamma_{i}$, since $\gamma_{i}>\frac{1}{2}$. For both biases to be negative, we need $\theta > \max\{\gamma_{i},1-\gamma_{i}\} = \gamma_{i}$, since $\gamma_{i}>\frac{1}{2}$. For the right-bias to be positive and the left-bias to be negative, we need $1-\gamma_{i}< \theta < \gamma_{i}$ to hold. The case in which the right bias is negative while the right-bias is positive never holds, since we assume $\gamma_{i}>\frac{1}{2}$. Therefore, we have the following summary.
\begin{enumerate}
\item if $\theta < 1- \gamma_{i}$, then both biases are strictly positive
\item if $1-\gamma_{i}<\theta<\gamma_{i}$, then right-bias is strictly positive and left-bias is strictly negative
\item if $\theta>\gamma_{i}$, then both biases are strictly negative.
\end{enumerate}
    
In the case (1) listed above, we say that the right-bias is less than the left bias whenever $\mu (\gamma_{i}-\theta) < \mu(1-\gamma_{i}-\theta)$, meaning that $\gamma_{i} < \frac{1}{2}$. However, this contradicts the assumption that individual is confirmatory and we can conclude that whenever $\theta < 1- \gamma_{i}$, the left-biased opinion is less biased than the right-biased one. In the case (3), we say that the right-bias is less than the left bias whenever $\mu (\theta-\gamma_{i}) < \mu(\gamma_{i}+\theta-1)$, meaning  that the statement is true if $\gamma_{i} > \frac{1}{2}$. Therefore, if $\theta>\gamma_{i}$, the right-biased opinion is less biased than the left-biased one. Finally, in the case (2), we say that the right-bias is less than the left bias whenever $\mu(\gamma_{i}-\theta)<\mu(\gamma_{i}+\theta-1)$, meaning that it can only be true when $\theta>\frac{1}{2}$. These three arguments together prove the statement and we conclude that the right-bias is less than the left bias whenever $\theta > \frac{1}{2}$ (and vice-versa).

Finally, when $\theta=\frac{1}{2}$, the biases are equal since $|\gamma_{i} - \frac{1}{2}|=|\frac{1}{2} - \gamma_{i}|$ for any $\gamma_{i}$. \hfill $\blacksquare$

\vspace{5mm}
\noindent \textbf{Proof of Corollary~\ref{cor:impartial}.} When an individual $j$ is always impartial, we have that
\begin{align}
\psi_{j,t} &= \frac{1}{2} \, \mathbbm{1}\{y_{j,t-1}\geq 0.5\} + \frac{1}{2} \, \mathbbm{1}\{y_{j,t-1} < 0.5\} \nonumber\\
			&= \frac{1}{2} \, \mathbbm{1}\{y_{j,t-1}\geq 0.5\} + \frac{1}{2} \, \left(1-\mathbbm{1}\{y_{j,t-1} \geq 0.5\}\right) \nonumber \\
			&= \frac{1}{2} \label{eq:psi-impartial},
\end{align} for all $t$. Since $u_{t}$ is a continuous $U\left[0,1\right]$ random variable in every period $t$, we have that
\begin{equation}
\label{eq:eta-impartial}
\mathbbm{E}_{t}\left[\mathbbm{1}\left\{u_{t} \leq \frac{1}{2}\right\}\right] = \mathbbm{P}\left(u_{t} \leq \frac{1}{2}\right) = F_{u}\left(\frac{1}{2}\right) = \frac{\frac{1}{2} - 0}{1-0} =  \frac{1}{2},
\end{equation} where $F_{u}(\cdot)$ is the cumulative distribution function of $U\left[0,1\right]$. Thus, equations \eqref{eq:singleagentproof} and \eqref{eq:eta-impartial} together prove the statement when agents are impartial (both always impartial and moderately impartial). \hfill $\blacksquare$

\vspace{5mm}
\noindent \textbf{Proof of Corollary~\ref{cor:extremism}.} Say extreme opinion 1 (i.e. $y_{i,\infty} =1$) is formed, then as per Propositions~\ref{prop:singleagent} and \ref{prop:t1ct2c_areas} we know this is the right-biased opinion and therefore it should be the case that $(1-\mu)\theta + \mu \gamma_{i} = 1$. Conversely, say extreme opinion 0 (i.e. $y_{i,\infty} = 0$) is formed. Then, we know this is the left-biased opinion and it should be that $(1-\mu)\theta + \mu (1-\gamma_{i}) = 0$. These two conditions together imply that $\mu (2\gamma_{i}-1)=1$. If we generally consider that $0\leq \mu \leq 1$ and $\frac{1}{2} \leq \gamma_{i} \leq 1$, then the relation $\mu (2\gamma_{i}-1)=1$ is only met when $\mu=\gamma_{i}=1$. \hfill $\blacksquare$

\vspace{5mm}
\noindent \textbf{Proof of Proposition~\ref{prop:ntw-society}.} As per Lemma \ref{lem:opinionint} in the Appendix C, the limiting opinion of any agent $i$ can be written as
\begin{align}
\lim_{t\to\infty} y_{i,t} & = \lim_{t\to\infty} \, \dfrac{\frac{1}{t} \sum_{j=1}^{n}W_{ij}^{t}\alpha_{j,0} + b \frac{1}{t} K(i,t)}{\frac{1}{t} \sum_{j=1}^{n}W_{ij}^{t}\left(\alpha_{j,0}+\beta_{j,0}\right) + b \frac{1}{t} L(i,t)} \nonumber \\
 		&= \lim_{t\to\infty} \, \dfrac{\dfrac{1}{t} \displaystyle \sum_{k=0}^{t-1}\sum_{j=1}^{n}W_{ij}^{k}s_{j,t-k}^{(1)}}{\dfrac{1}{t} \displaystyle \sum_{k=0}^{t-1} \sum_{j=1}^{n} W_{ij}^{k} \left(s_{j,t-k}^{(0)} + s_{j,t-k}^{(1)}\right)}. \nonumber
\end{align}

By Lemma \ref{lem:ergochain} we can split both series in the numerator and denominator in two parts 	 	 	
\begin{align}
 \lim_{t\to\infty} y_{i,t} &= \lim_{t\to\infty} \, \dfrac{\dfrac{1}{t} \left( \displaystyle \sum_{k=0}^{t_{\text{mix}}}\sum_{j=1}^{n}W_{ij}^{k}s_{j,t-k}^{(1)} + \sum_{k=t_{\text{mix}}+1}^{t-1}\sum_{j=1}^{n}W_{ij}^{k}s_{j,t-k}^{(1)} \right)}{\dfrac{1}{t} \displaystyle \left( \sum_{k=0}^{t_{\text{mix}}} \sum_{j=1}^{n} W_{ij}^{k} \left(s_{j,t-k}^{(0)} + s_{j,t-k}^{(1)}\right) + \sum_{k=t_{\text{mix}}+1}^{t-1} \sum_{j=1}^{n} W_{ij}^{k} \left(s_{j,t-k}^{(0)} + s_{j,t-k}^{(1)}\right)\right)} \nonumber \\
	&= \lim_{t\to\infty} \, \dfrac{\dfrac{1}{t}  \displaystyle \sum_{k=t_{\text{mix}}+1}^{t-1}\sum_{j=1}^{n}W_{ij}^{k}s_{j,t-k}^{(1)}}{\dfrac{1}{t} \displaystyle \sum_{k=t_{\text{mix}}+1}^{t-1} \sum_{j=1}^{n} W_{ij}^{k} \left(s_{j,t-k}^{(0)} + s_{j,t-k}^{(1)}\right)}. \nonumber
\end{align} 		

Since the subindex $k$ spans from $t_{mix}$ onwards (i.e. when the chain is already mixed), we can use the invariant distribution matrix in the previous expression. Therefore the limiting opinion becomes
\begin{align}
\lim_{t\to\infty} y_{i,t} 		&= \lim_{t\to\infty} \, \dfrac{\displaystyle \sum_{j=1}^{n}\Pi_{ij} \, \dfrac{1}{t} \sum_{k=t_{\text{mix}}+1}^{t-1} s_{j,t-k}^{(1)}}{\displaystyle  \sum_{j=1}^{n} \Pi_{ij} \, \dfrac{1}{t} \sum_{k=t_{\text{mix}}+1}^{t-1} \left(s_{j,t-k}^{(0)} + s_{j,t-k}^{(1)}\right)} \nonumber \\
 		&= \dfrac{\displaystyle \sum_{j=1}^{n}\Pi_{ij} \, \lim_{t\to\infty} \dfrac{t-1-t_{\text{mix}}}{t} \dfrac{1}{t-1-t_{\text{mix}}} \sum_{k=t_{\text{mix}}+1}^{t-1} s_{j,t-k}^{(1)}}{\displaystyle  \sum_{j=1}^{n} \Pi_{ij} \, \lim_{t\to\infty} \dfrac{t-1-t_{\text{mix}}}{t} \dfrac{1}{t-1-t_{\text{mix}}} \sum_{k=t_{\text{mix}}+1}^{t-1} \left(s_{j,t-k}^{(0)} + s_{j,t-k}^{(1)}\right)} \nonumber \\
 		&= \dfrac{\displaystyle \sum_{j=1}^{n}\Pi_{ij} \, \lim_{t\to\infty} \dfrac{1}{t-1-t_{\text{mix}}} \sum_{k=t_{\text{mix}}+1}^{t-1} \left(\mathbbm{1}\{s_{t-k}=1\} +  \mathbbm{1}\{s_{t-k}=a\}\mathbbm{1}\{u_{t-k} \leq \psi_{j,t-k}\}\right)}{\displaystyle  \sum_{j=1}^{n} \Pi_{ij} \, \lim_{t\to\infty} \dfrac{1}{t-1-t_{\text{mix}}} \sum_{k=t_{\text{mix}}+1}^{t-1} \left(\mathbbm{1}\{s_{t-k}=0\} + \mathbbm{1}\{s_{t-k}=1\} +  \mathbbm{1}\{s_{t-k}=a\}\right)} \nonumber\\
 		&= \dfrac{\sum_{j} \Pi_{ij} \mathbbm{E}_{t}\left[\mathbbm{1}\{s_{t}=1\} +  \mathbbm{1}\{s_{t}=a\}\mathbbm{1}\{u_{t} \leq \psi_{j,t}\}\right]}{\sum_{j} \Pi_{ij} \mathbbm{E}_{t}\left[\mathbbm{1}\{s_{t}=0\} + \mathbbm{1}\{s_{t}=1\} +  \mathbbm{1}\{s_{t}=a\}\right]} \nonumber\\
 		&= (1-\mu)\theta + \mu \sum_{j}\Pi_{ij} \mathbbm{E}_{t}\left[\mathbbm{1}\{u_{t} \leq \psi_{j,t}\}\right], \nonumber
\end{align} where the term $\mathbbm{E}_{t}\left[\mathbbm{1}\{u_{t} \leq \psi_{j,t}\}\right]$ is as in Proposition~\ref{prop:singleagent}, implying that the limiting consensus is
\[
 \lim_{t\to\infty} y_{i,t} = \begin{cases} 
      (1-\mu)\theta + \mu \sum_{j}\Pi_{ij}\gamma_{j} & , \text{if} \,\, A \nonumber \\
     (1-\mu)\theta + \mu \sum_{j}\Pi_{ij}(1-\gamma_{j}) & , \text{if} \,\, B 
  \end{cases} 
\] \hfill $\blacksquare$

\vspace{5mm}
\noindent \textbf{Proof of Proposition~\ref{prop:b-zero}.} From Equation \eqref{eq:alphaupalgfwd} in Appendix C, we know that $\alpha_{i,t}$, for any $i$, can be iterated forwardly as
\begin{equation*}
\alpha_{i,t} = \sum_{j=1}^{n}W_{ij}^{t}\alpha_{j,0} + b \, \sum_{k=0}^{t-1}\sum_{j=1}^{n} W_{ij}^{k} s_{j,t-k}^{(1)}.
\end{equation*}.

Similarly, the expression $\alpha_{i,t} + \beta_{i,t}$ in Equation \eqref{eq:alphaplusbetaalgfwd} can be written as
\begin{equation*}
\alpha_{i,t} + \beta_{i,t} = \sum_{j=1}^{n}W_{ij}^{t} \left(\alpha_{j,0} + \beta_{j,0}\right) + b \, \sum_{k=0}^{t-1}\sum_{j=1}^{n} W_{ij}^{k} \left(s_{j,t-k}^{(0)}+s_{j,t-k}^{(1)}\right).
\end{equation*}

Thus, if $b=0$, the opinion of any agent $i \in N$ at any time $t$ boils down to
\[ 
y_{i,t} = \frac{\sum_{j=1}^{n}W_{ij}^{t}\alpha_{j,0}}{\sum_{j=1}^{n}W_{ij}^{t} \left(\alpha_{j,0} + \beta_{j,0}\right)}
\] 
and therefore
\[
\lim_{t\to\infty} y_{i,t} = y = \frac{\sum_{j=1}^{n}\Pi_{ij}\alpha_{j,0}}{\sum_{j=1}^{n}\Pi_{ij} \left(\alpha_{j,0} + \beta_{j,0}\right)}
\] for any $i$. 
In this case, the limiting opinion of any agent $i$ can be written as in the case when $b=0$ shown above. \hfill $\blacksquare$

\pagebreak
\section{Tests concerning differences among proportions}
\label{append:proptests}

\subsection{Definition}
\label{append-subsec:def}
To decide whether observed differences among sample proportions are significant or whether they can be attributed to chance we must use tests concerning differences among proportions. For that, suppose that $x_1, x_2, \dots , x_k$ are observed values of $k$ independent random variables $X_1,X_2, . . . ,X_k$
having binomial distributions with the parameters $n_1$ and $\theta_1$, $n_2$ and $\theta_2, \dots$, $n_k$ and $\theta_{k}$. If the sample sizes are sufficiently large, we can approximate the distributions of the independent
random variables
\[
Z_{i} = \dfrac{X_{i} - n_{i}\theta_{i}}{\sqrt{n_{i}\theta_{i}(1-\theta_{i})}} \hspace{5mm} \text{ for } i=1,2,\dots,k
\] with standard normal distributions. Therefore, we know that we can look upon the test-statistic
\[
\chi^{2} = \sum_{i=1}^{k} Z_{i}^{2}  = \sum_{i=1}^{k} \dfrac{(x_{i}-n_{i}\theta_{i})^2}{n_{i}\theta_{i}(1-\theta_{i})}
\] as a value of a random variable having chi-square distribution with $k$ degrees of freedom. When the null hypothesis $H_{0}$ is $\theta_{1} = \theta_{2} = \cdots = \theta_{k}$ and the alternative hypothesis is that at least one of the $\theta$'s is different, we can use the \textit{pooled estimate}
\[
\hat{\theta} = \dfrac{\sum_{i=1}^{k}x_{i}}{\sum_{i=1}^{k}n_{i}}
\] and the test statistic becomes
\[
\chi^{2} = \sum_{i=1}^{k} \dfrac{(x_{i}-n_{i}\hat{\theta})^2}{n_{i}\hat{\theta}(1-\hat{\theta})}
\] a random variable whose value has chi-square distribution with $k-1$ degrees of freedom because an estimate is substituted for the unknown parameter $\theta$.

\input{table.mat.01}

\input{teste.tex}

\pagebreak
\section{Probit regression model - Robustness}
\label{append:sims-stats}
\input{probit_pool.tex}


\end{document}

%% file: graph-line-n3.tex
\begin{tikzpicture}
  [scale=1.5,auto=left,every node/.style={circle,fill=gray!30,draw=black,line width=0.8pt,scale=0.8}]
  \node (n1) at (1,2)  {1};
  \node (n2) at (2,2)  {2};
  \node (n3) at (3,2)  {3};
  
      \foreach \from/\to in {n1/n2, n2/n3}
    \draw[line width=0.22mm] (\from) -- (\to);
\end{tikzpicture}

%% file: graph-circle-n3.tex
\begin{tikzpicture}
  [scale=1.5,auto=left,every node/.style={circle,fill=gray!30,draw=black,line width=0.8pt,scale=0.8}]
  \node (n1) at (1,1)  {1};
  \node (n2) at (2,2.3)  {2};
  \node (n3) at (3,1)  {3};
  
      \foreach \from/\to in {n1/n2,n1/n3,n2/n3}
    \draw [line width=0.22mm] (\from) -- (\to);
\end{tikzpicture}

%% file: netA.tex
\begin{tikzpicture}
  [scale=0.6,auto=left,every node/.style={circle ,fill=black!100,scale=0.7}]
  \node (n1) at (1,2)  {};
  \node (n2) at (2,2)  {};

    \foreach \from/\to in {n1/n2}
    \draw (\from) -- (\to);
\end{tikzpicture}

%% file: netB.tex
\begin{tikzpicture}
  [scale=0.6,auto=left,every node/.style={circle,fill=black!100,scale=0.7}]
  \node (n1) at (1,2)  {};
  \node (n2) at (2,2)  {};
  \node (n3) at (3,2)  {};
  
      \foreach \from/\to in {n1/n2, n2/n3}
    \draw (\from) -- (\to);
\end{tikzpicture}

%% file: netC.tex
\begin{tikzpicture}
  [scale=0.6,auto=left,every node/.style={circle,fill=black!100,scale=0.7}]
  \node (n1) at (1,1)  {};
  \node (n2) at (2,2.3)  {};
  \node (n3) at (3,1)  {};
  
      \foreach \from/\to in {n1/n2,n1/n3,n2/n3}
    \draw (\from) -- (\to);
\end{tikzpicture}

%% file: netD.tex
\begin{tikzpicture}
  [scale=0.6,auto=left,every node/.style={circle,fill=black!100,scale=0.7}]
  \node (n1) at (1,2)  {};
  \node (n2) at (2,2)  {};
  \node (n3) at (3,2)  {};
  \node (n4) at (4,2)  {};
  
      \foreach \from/\to in {n1/n2,n2/n3,n3/n4}
    \draw (\from) -- (\to);
\end{tikzpicture}

%% file: netE.tex
\begin{tikzpicture}
  [scale=0.6,auto=left,every node/.style={circle,fill=black!100,scale=0.7}]
  \node (n1) at (1,1)  {};
  \node (n2) at (2,2)  {};
  \node (n3) at (2,3.3)  {};
  \node (n4) at (3,1)  {};

  \foreach \from/\to in {n1/n2,n2/n3,n2/n4}
    \draw (\from) -- (\to);
\end{tikzpicture}

%% file: netF.tex
\begin{tikzpicture}
  [scale=0.6,auto=left,every node/.style={circle,fill=black!100,scale=0.7}]
  \node (n1) at (1,2)  {};
  \node (n2) at (2,1)  {};
  \node (n3) at (3,2)  {};
  \node (n4) at (2,3)  {};

  \foreach \from/\to in {n1/n2,n2/n3,n3/n4,n4/n1}
    \draw (\from) -- (\to);
\end{tikzpicture}

%% file: netG.tex
\begin{tikzpicture}
  [scale=0.6,auto=left,every node/.style={circle,fill=black!100,scale=0.7}]
  \node (n1) at (1,2)  {};
  \node (n2) at (2,1)  {};
  \node (n3) at (3,2)  {};
  \node (n4) at (2,3)  {};

  \foreach \from/\to in {n1/n2,n2/n3,n3/n4,n4/n1,n1/n3,n2/n4}
    \draw (\from) -- (\to);
\end{tikzpicture}

%% file: netH.tex
\begin{tikzpicture}
  [scale=0.6,auto=left,every node/.style={circle,fill=black!100,scale=0.7}]
  \node (n1) at (1,2)  {};
  \node (n2) at (2.3,2)  {};
  \node (n3) at (3.6,1)  {};
  \node (n4) at (3.6,3)  {};

  \foreach \from/\to in {n1/n2,n2/n3,n2/n4,n3/n4}
    \draw (\from) -- (\to);
\end{tikzpicture}

%% file: stats_0208.tex
\begin{table}[ht] 
\centering 
\footnotesize 
\begin{tabular}{@{\extracolsep{5pt}}lcccccccc} 
\\[-1.8ex]\hline 
\hline \\[-1.8ex] 
Statistic & \multicolumn{1}{c}{N} & \multicolumn{1}{c}{Mean} & \multicolumn{1}{c}{St. Dev.} & \multicolumn{1}{c}{Min} & \multicolumn{1}{c}{Pctl(25)} & \multicolumn{1}{c}{Median} & \multicolumn{1}{c}{Pctl(75)} & \multicolumn{1}{c}{Max} \\ 
\hline \\[-1.8ex] 
$\hat{p}_{(A)}$ & 21,040 & 0.809 & 0.393 & 0 & 1 & 1 & 1 & 1 \\ 
$\hat{p}_{(B)}$ & 21,040 & 0.679 & 0.467 & 0 & 0 & 1 & 1 & 1 \\ 
$\hat{p}_{(C)}$ & 21,040 & 0.810 & 0.392 & 0 & 1 & 1 & 1 & 1 \\ 
$\hat{p}_{(D)}$ & 21,040 & 0.735 & 0.441 & 0 & 0 & 1 & 1 & 1 \\ 
$\hat{p}_{(E)}$ & 21,040 & 0.698 & 0.459 & 0 & 0 & 1 & 1 & 1 \\ 
$\hat{p}_{(F)}$ & 21,040 & 0.826 & 0.379 & 0 & 1 & 1 & 1 & 1 \\ 
$\hat{p}_{(G)}$ & 21,040 & 0.787 & 0.409 & 0 & 1 & 1 & 1 & 1 \\ 
$\hat{p}_{(H)}$ & 21,040 & 0.700 & 0.458 & 0 & 0 & 1 & 1 & 1 \\ 
$\mathcal{R}_{0}$ degree in (B) & 21,040 & 1.332 & 0.471 & 1 & 1 & 1 & 2 & 2 \\ 
$\mathcal{R}_{0}$ degree in (D) & 21,040 & 1.503 & 0.500 & 1 & 1 & 2 & 2 & 2 \\ 
$\mathcal{R}_{0}$ degree in (E) & 21,040 & 1.516 & 0.875 & 1 & 1 & 1 & 3 & 3 \\ 
$\mathcal{R}_{0}$ degree in (H) & 21,040 & 2.010 & 0.704 & 1 & 2 & 2 & 3 & 3 \\ 
$\mathcal{L}_{0}$ degree in (B) & 21,040 & 1.333 & 0.471 & 1 & 1 & 1 & 2 & 2 \\ 
$\mathcal{L}_{0}$ degree in (D) & 21,040 & 1.503 & 0.500 & 1 & 1 & 2 & 2 & 2 \\ 
$\mathcal{L}_{0}$ degree in (E) & 21,040 & 1.501 & 0.866 & 1 & 1 & 1 & 3 & 3 \\ 
$\mathcal{L}_{0}$ degree in (H) & 21,040 & 1.996 & 0.707 & 1 & 1 & 2 & 2 & 3 \\ 
Open mind (OM) in (B) & 21,040 & 0.666 & 0.472 & 0 & 0 & 1 & 1 & 1 \\ 
Open mind (OM) in (D) & 21,040 & 0.501 & 0.500 & 0 & 0 & 1 & 1 & 1 \\ 
Open mind (OM) in (E) & 21,040 & 0.508 & 0.500 & 0 & 0 & 1 & 1 & 1 \\ 
Open mind (OM) in (H) & 21,040 & 0.667 & 0.471 & 0 & 0 & 1 & 1 & 1 \\ 
First impression (FI) & 21,040 & 0.619 & 0.486 & 0 & 0 & 1 & 1 & 1 \\ 
$\mathbbm{1}\{\theta=0.8\}$ & 21,040 & 0.500 & 0.500 & 0 & 0 & 0.5 & 1 & 1 \\ 
$\mathbbm{1}\{\tau=0\}$ & 21,040 & 0.250 & 0.433 & 0 & 0 & 0 & 0.2 & 1 \\ 
$\mathbbm{1}\{\tau=1\}$ & 21,040 & 0.250 & 0.433 & 0 & 0 & 0 & 0.2 & 1 \\ 
$\mathbbm{1}\{\tau=10\}$ & 21,040 & 0.250 & 0.433 & 0 & 0 & 0 & 0.2 & 1 \\ 
$\mathbbm{1}\{\tau=30\}$ & 21,040 & 0.250 & 0.433 & 0 & 0 & 0 & 0.2 & 1 \\ 
\hline \\[-1.8ex] 
\end{tabular} 
\caption{Summary statistics - simulated $\hat{p}$ and parameters} 
\label{tab:sumstatsim} 
\end{table} 

%% file: probit_bench.tex
\begin{table}[ht] \centering 
\scriptsize 
\begin{tabular}{@{\extracolsep{-15pt}}lD{.}{.}{-2} D{.}{.}{-2} D{.}{.}{-2} D{.}{.}{-2} D{.}{.}{-2} D{.}{.}{-2} D{.}{.}{-2} D{.}{.}{-2} } 
\\[-1.8ex]\hline 
\hline \\[-1.8ex] 
 & \multicolumn{8}{c}{Dep. Variable: probability of emergence of less biased consensus ($\hat{p}_G$)} \\ 
\cline{2-9} 
 & \multicolumn{1}{c}{(A)} & \multicolumn{1}{c}{(B)} & \multicolumn{1}{c}{(C)} & \multicolumn{1}{c}{(D)} & \multicolumn{1}{c}{(E)} & \multicolumn{1}{c}{(F)} & \multicolumn{1}{c}{(G)} & \multicolumn{1}{c}{(H)} \\ 
\hline \\[-1.8ex] 
 Partisan centrality advantage (PCA) &  & 1.86^{***} &  & 0.87^{***} & 1.91^{***} &  &  & 1.62^{***} \\ 
  &  & (0.04) &  & (0.04) & (0.05) &  &  & (0.05) \\ 
  Open mind (OM) &  & -1.06^{***} &  & 0.003 & -0.96^{***} & 0.50^{***} &  & 0.45^{***} \\ 
  &  & (0.04) &  & (0.03) & (0.04) & (0.03) &  & (0.04) \\ 
  First impression (FI) & 1.64^{***} & 2.13^{***} & 1.93^{***} & 1.13^{***} & 2.24^{***} & 2.32^{***} & 2.16^{***} & 1.36^{***} \\ 
  & (0.03) & (0.05) & (0.04) & (0.04) & (0.05) & (0.06) & (0.04) & (0.05) \\ 
  PCA $\times$ FI &  & 0.83^{***} &  & 0.91^{***} & 0.77^{***} &  &  & 0.61^{***} \\ 
  &  & (0.08) &  & (0.08) & (0.10) &  &  & (0.07) \\ 
  PCA $\times$ OM &  &  &  & -0.05 &  &  &  & -0.39^{***} \\ 
  &  &  &  & (0.06) &  &  &  & (0.06) \\ 
  OM $\times$ FI &  & -1.10^{***} &  & -0.29^{***} & -1.18^{***} & -0.61^{***} &  & -0.10^{*} \\ 
  &  & (0.06) &  & (0.04) & (0.06) & (0.06) &  & (0.05) \\ 
  $\mathbbm{1}\{\tau=1\}$ & 0.54^{***} & 0.48^{***} & 0.56^{***} & 0.58^{***} & 0.34^{***} & 0.52^{***} & 0.33^{***} & 0.42^{***} \\ 
  & (0.03) & (0.03) & (0.03) & (0.03) & (0.03) & (0.03) & (0.03) & (0.03) \\ 
  $\mathbbm{1}\{\tau=10\}$ & 0.54^{***} & -0.41^{***} & 0.66^{***} & 0.13^{***} & -0.37^{***} & 0.87^{***} & 0.38^{***} & -0.08^{*} \\ 
  & (0.03) & (0.03) & (0.03) & (0.03) & (0.03) & (0.04) & (0.03) & (0.03) \\ 
  $\mathbbm{1}\{\tau=30\}$ & 0.54^{***} & -0.41^{***} & 0.66^{***} & -0.16^{***} & -0.38^{***} & 0.98^{***} & 0.39^{***} & -0.65^{***} \\ 
  & (0.03) & (0.03) & (0.03) & (0.03) & (0.03) & (0.04) & (0.03) & (0.03) \\ 
  $\mathbbm{1}\{\theta=0.8\}$ & -0.33^{***} & -0.02 & -0.26^{***} & -0.05 & 0.04 & -0.15^{***} & -0.19^{***} & -0.06^{*} \\ 
  & (0.03) & (0.03) & (0.04) & (0.03) & (0.03) & (0.04) & (0.04) & (0.03) \\ 
  Constant & -0.15^{***} & 0.03 & -0.31^{***} & -0.27^{***} & -0.08^{**} & -0.68^{***} & -0.30^{***} & -0.85^{***} \\ 
  & (0.02) & (0.03) & (0.02) & (0.03) & (0.03) & (0.03) & (0.03) & (0.04) \\ 
 \hline \\[-1.8ex] 
Observations & \multicolumn{1}{c}{21,040} & \multicolumn{1}{c}{21,040} & \multicolumn{1}{c}{21,040} & \multicolumn{1}{c}{21,040} & \multicolumn{1}{c}{21,040} & \multicolumn{1}{c}{21,040} & \multicolumn{1}{c}{21,040} & \multicolumn{1}{c}{21,040} \\ 
Log Likelihood & \multicolumn{1}{c}{-7,873.93} & \multicolumn{1}{c}{-7,648.32} & \multicolumn{1}{c}{-7,053.26} & \multicolumn{1}{c}{-9,546.68} & \multicolumn{1}{c}{-7,350.89} & \multicolumn{1}{c}{-6,348.13} & \multicolumn{1}{c}{-6,941.64} & \multicolumn{1}{c}{-8,660.61} \\ 
Akaike Inf. Crit. & \multicolumn{1}{c}{15,759.90} & \multicolumn{1}{c}{15,316.60} & \multicolumn{1}{c}{14,118.50} & \multicolumn{1}{c}{19,115.40} & \multicolumn{1}{c}{14,721.80} & \multicolumn{1}{c}{12,712.30} & \multicolumn{1}{c}{13,895.30} & \multicolumn{1}{c}{17,343.20} \\ 
\hline 
\hline \\[-1.8ex] 
\textit{Note:}  & \multicolumn{8}{r}{$^{*}$p$<$0.05; $^{**}$p$<$0.01; $^{***}$p$<$0.001} \\ 
\end{tabular} 
  \caption{Regression results: Probit} 
  \label{tab:probit-bench} 
\end{table} 

%% file: netSA.tex
\begin{tikzpicture}
  [scale=0.6,auto=left,every node/.style={circle ,fill=black!100,scale=0.7}]
  \node (n1) at (1,2)  {};

\end{tikzpicture}

%% file: open-narrow.tex
\begin{table}[ht]
\centering
\begingroup\footnotesize
\begin{tabular}{llcccc}
  \toprule
Network & Partisans & $\underset{(\tau=0)}{\hat{p}}$ & $\underset{(\tau=1)}{\hat{p}}$ & $\underset{(\tau=10)}{\hat{p}}$ & $\underset{(\tau=30)}{\hat{p}}$ \\ 
  \midrule
(B) & pooled & 0.702 & 0.802 & 0.603 & 0.61 \\ 
   & open-minded & 0.706 & 0.806 & 0.493 & 0.502 \\ 
   & norrow-minded & 0.693 & 0.794 & 0.821 & 0.826 \\ 
   \midrule
(D) & pooled & 0.702 & 0.841 & 0.739 & 0.659 \\ 
   & open-minded & 0.693 & 0.846 & 0.668 & 0.649 \\ 
   & norrow-minded & 0.711 & 0.837 & 0.812 & 0.668 \\ 
   \midrule
(E) & pooled & 0.721 & 0.787 & 0.643 & 0.641 \\ 
   & open-minded & 0.724 & 0.797 & 0.507 & 0.498 \\ 
   & norrow-minded & 0.719 & 0.776 & 0.782 & 0.784 \\ 
   \midrule
(F) & pooled & 0.702 & 0.818 & 0.884 & 0.899 \\ 
   & open-minded & 0.713 & 0.829 & 0.918 & 0.936 \\ 
   & norrow-minded & 0.681 & 0.796 & 0.812 & 0.825 \\ 
   \midrule
(H) & pooled & 0.72 & 0.814 & 0.706 & 0.558 \\ 
   & open-minded & 0.719 & 0.82 & 0.711 & 0.579 \\ 
   & norrow-minded & 0.723 & 0.803 & 0.695 & 0.517 \\ 
   \bottomrule
\end{tabular}
\endgroup
\caption{Open and norrow-minded partisans - Result~\ref{res:open-narrow}} 
\label{tab:open-narrow}
\end{table}

%% file: table.mat.01.tex
\begin{table}[H]
\centering
\begingroup\small
\begin{tabular}{cccccccccc}
  \hline
$i$ & $j$ & $c_{i}$ & $c_{j}$ & $\hat{p}_{i}(c_{i})$ & $\hat{p}_{j}(c_{j})$ & $CI_{5\%}$ & $CI_{95\%}$ & $\chi^{2}$ & p-value \\ 
  \hline
(A) & (E) & $\tau=0$ & $\tau=0$ & 0.702 & 0.721 & -0.037 & -0.002 & 4.915 & 0.027 \\ 
  (A) & (G) & $\tau=0$ & $\tau=0$ & 0.702 & 0.727 & -0.042 & -0.007 & 7.871 & 0.005 \\ 
  (A) & (H) & $\tau=0$ & $\tau=0$ & 0.702 & 0.72 & -0.035 & -0.001 & 4.174 & 0.041 \\ 
  (E) & (G) & $\tau=0$ & $\tau=0$ & 0.721 & 0.727 & -0.022 & 0.012 & 0.347 & 0.556 \\ 
  (E) & (H) & $\tau=0$ & $\tau=0$ & 0.721 & 0.72 & -0.016 & 0.019 & 0.03 & 0.862 \\ 
  (G) & (H) & $\tau=0$ & $\tau=0$ & 0.727 & 0.72 & -0.01 & 0.024 & 0.582 & 0.446 \\ 
   \hline
\end{tabular}
\endgroup
\caption{Hypothesis Test for Proportions - Result 1} 
\end{table}

%% file: teste.tex
\begin{table}[H]
\centering

\begingroup\small
\begin{tabular}{cccccccccc}
  \hline
$i$ & $j$ & $c_{i}$ & $c_{j}$ & $\hat{p}_{i}(c_{i})$ & $\hat{p}_{j}(c_{j})$ & $CI_{5\%}$ & $CI_{95\%}$ & $\chi^{2}$ & p-value \\ 
  \hline
(A) & (A) & $\tau=0$ & $\tau=30$ & 0.688 & 0.678 & -0.006 & 0.026 & 1.641 & 0.2 \\ 
  (B) & (B) & $\tau=0$ & $\tau=30$ & 0.688 & 0.59 & 0.082 & 0.115 & 140.542 & 0 \\ 
  (D) & (D) & $\tau=0$ & $\tau=30$ & 0.688 & 0.648 & 0.024 & 0.056 & 23.694 & 0 \\ 
  (B) & (D) & $\tau=1$ & $\tau=1$ & 0.707 & 0.766 & -0.077 & -0.041 & 41.405 & 0 \\ 
  (B) & (D) & $\tau=30$ & $\tau=30$ & 0.59 & 0.648 & -0.078 & -0.039 & 32.948 & 0 \\ 
   \hline
\end{tabular}
\endgroup
\caption{Two Population Proportions - Result~\ref{res:linenets}} 
\end{table}

%% file: probit_pool.tex
\begin{table}[ht] \centering  
\footnotesize 
\begin{tabular}{@{\extracolsep{-20pt}}lD{.}{.}{-2} D{.}{.}{-2} D{.}{.}{-2} } 
\\[-1.8ex]\hline 
\hline \\[-1.8ex] 
 & \multicolumn{3}{c}{Dep. Variable: probability of emergence of less biased consensus ($\hat{p}_G$)} \\ 
\cline{2-4} 
 & \multicolumn{1}{c}{Pooled} & \multicolumn{1}{c}{Pooled ($\theta=0.2$)} & \multicolumn{1}{c}{Pooled ($\theta=0.8$)} \\ 
\hline \\[-1.8ex] 
 Partisan centrality advantage (PCA) & 1.05^{***}$ $(0.03) & 1.06^{***}$ $(0.03) & 1.23^{***}$ $(0.08) \\ 
  Open mind (OM) & -0.40^{***}$ $(0.01) & -0.25^{***}$ $(0.02) & -0.93^{***}$ $(0.04) \\ 
  First impression (FI) & 1.71^{***}$ $(0.02) & 1.63^{***}$ $(0.03) & 1.85^{***}$ $(0.04) \\ 
  PCA $\times$ FI & 0.37^{***}$ $(0.04) & 0.18^{**}$ $(0.06) & 0.37^{***}$ $(0.07) \\ 
  PCA $\times$ OM & 0.52^{***}$ $(0.03) & 0.41^{***}$ $(0.04) & 0.76^{***}$ $(0.08) \\ 
  OM $\times$ FI & -0.29^{***}$ $(0.02) & -0.27^{***}$ $(0.03) & 0.02$ $(0.05) \\ 
  $\mathbbm{1}\{\tau=1\}$ & 0.46^{***}$ $(0.01) & 0.84^{***}$ $(0.01) & -0.51^{***}$ $(0.02) \\ 
  $\mathbbm{1}\{\tau=10\}$ & 0.19^{***}$ $(0.01) & 0.75^{***}$ $(0.01) & -1.06^{***}$ $(0.02) \\ 
  $\mathbbm{1}\{\tau=30\}$ & 0.07^{***}$ $(0.01) & 0.68^{***}$ $(0.01) & -1.23^{***}$ $(0.02) \\ 
  $\mathbbm{1}\{n=2\}$ & 0.47^{***}$ $(0.02) & -0.02$ $(0.02) & 1.47^{***}$ $(0.05) \\ 
  $\mathbbm{1}\{n=3\}$ & 0.49^{***}$ $(0.02) & -0.09^{***}$ $(0.02) & 1.68^{***}$ $(0.05) \\ 
  $\mathbbm{1}\{n=4\}$ & 0.40^{***}$ $(0.02) & -0.25^{***}$ $(0.02) & 1.74^{***}$ $(0.05) \\ 
  $\mathbbm{1}\{G=(B)\}$ & -1.15^{***}$ $(0.02) & -0.94^{***}$ $(0.02) & -1.57^{***}$ $(0.03) \\ 
  $\mathbbm{1}\{G=(D)\}$ & -0.84^{***}$ $(0.02) & -0.52^{***}$ $(0.02) & -1.44^{***}$ $(0.03) \\ 
  $\mathbbm{1}\{G=(E)\}$ & -0.96^{***}$ $(0.02) & -0.73^{***}$ $(0.02) & -1.44^{***}$ $(0.03) \\ 
  $\mathbbm{1}\{G=(F)\}$ & 0.02$ $(0.02) & 0.13^{***}$ $(0.02) & -0.13^{***}$ $(0.03) \\ 
  $\mathbbm{1}\{G=(H)\}$ & -1.02^{***}$ $(0.02) & -0.76^{***}$ $(0.02) & -1.54^{***}$ $(0.03) \\ 
  $\mathbbm{1}\{\theta=0.8\}$ & -0.10^{***}$ $(0.01) &  &  \\ 
 \hline \\[-1.8ex] 
Observations & \multicolumn{1}{c}{168,320} & \multicolumn{1}{c}{84,160} & \multicolumn{1}{c}{84,160} \\ 
Log Likelihood & \multicolumn{1}{c}{-67,244.40} & \multicolumn{1}{c}{-41,815.30} & \multicolumn{1}{c}{-21,296.40} \\ 
Akaike Inf. Crit. & \multicolumn{1}{c}{134,525.00} & \multicolumn{1}{c}{83,664.70} & \multicolumn{1}{c}{42,626.90} \\ 
\hline 
\hline \\[-1.8ex] 
\textit{Note:}  & \multicolumn{3}{r}{$^{*}$p$<$0.05; $^{**}$p$<$0.01; $^{***}$p$<$0.001} \\ 
\end{tabular} 
  \caption{Probit regression with pooled data} 
  \label{reg:probit-pool}
\end{table}

%% file: CBSN_draft_3rdRR_MSS_v01.bbl
\begin{thebibliography}{62}
\newcommand{\enquote}[1]{``#1''}
\expandafter\ifx\csname natexlab\endcsname\relax\def\natexlab#1{#1}\fi

\bibitem[\protect\citeauthoryear{Acemoglu, Bimpikis, and Ozdaglar}{Acemoglu
  et~al.}{2014}]{abo2014}
\textsc{Acemoglu, D., K.~Bimpikis, and A.~Ozdaglar} (2014): \enquote{Dynamics
  of information exchange in endogenous social networks,} \emph{Theoretical
  Economics}, 9, 41--97.

\bibitem[\protect\citeauthoryear{Acemoglu, Dahleh, Lobel, and
  Ozdaglar}{Acemoglu et~al.}{2011}]{adlo2011}
\textsc{Acemoglu, D., M.~A. Dahleh, I.~Lobel, and A.~Ozdaglar} (2011):
  \enquote{Bayesian learning in social networks,} \emph{The Review of Economic
  Studies}, 78, 1201--1236.

\bibitem[\protect\citeauthoryear{Acemoglu and Ozdaglar}{Acemoglu and
  Ozdaglar}{2011}]{Acemoglu2011}
\textsc{Acemoglu, D. and A.~Ozdaglar} (2011): \enquote{Opinion dynamics and
  learning in social networks,} \emph{Dynamic Games and Applications}, 1,
  3--49.

\bibitem[\protect\citeauthoryear{Acemoglu, Ozdaglar, and
  ParandehGheibi}{Acemoglu et~al.}{2010}]{aceozdpar2010}
\textsc{Acemoglu, D., A.~Ozdaglar, and A.~ParandehGheibi} (2010):
  \enquote{Spread of (mis) information in social networks,} \emph{Games and
  Economic Behavior}, 70, 194--227.

\bibitem[\protect\citeauthoryear{Allahverdyan and Galstyan}{Allahverdyan and
  Galstyan}{2014}]{Allahverdyan2014}
\textsc{Allahverdyan, A.~E. and A.~Galstyan} (2014): \enquote{Opinion dynamics
  with confirmation bias,} \emph{PloS one}, 9, e99557.

\bibitem[\protect\citeauthoryear{Andreoni and Mylovanov}{Andreoni and
  Mylovanov}{2012}]{Andreoni2012}
\textsc{Andreoni, J. and T.~Mylovanov} (2012): \enquote{Diverging opinions,}
  \emph{American Economic Journal: Microeconomics}, 4, 209--32.

\bibitem[\protect\citeauthoryear{Andrews, Logan, and Sinkey}{Andrews
  et~al.}{2018}]{andrews2018identifying}
\textsc{Andrews, R.~J., T.~D. Logan, and M.~J. Sinkey} (2018):
  \enquote{Identifying confirmatory bias in the field: Evidence from a poll of
  experts,} \emph{Journal of Sports Economics}, 19, 50--81.

\bibitem[\protect\citeauthoryear{Azzimonti and Fernandes}{Azzimonti and
  Fernandes}{2022}]{AzzFer2018}
\textsc{Azzimonti, M. and M.~Fernandes} (2022): \enquote{Social media networks,
  fake news, and polarization,} \emph{European Journal of Political Economy},
  102256.

\bibitem[\protect\citeauthoryear{Bala and Goyal}{Bala and
  Goyal}{1998}]{balagoyal1998}
\textsc{Bala, V. and S.~Goyal} (1998): \enquote{Learning from neighbours,}
  \emph{The review of economic studies}, 65, 595--621.

\bibitem[\protect\citeauthoryear{Bala and Goyal}{Bala and
  Goyal}{2001}]{Bala2001}
---\hspace{-.1pt}---\hspace{-.1pt}--- (2001): \enquote{Conformism and diversity
  under social learning,} \emph{Economic theory}, 17, 101--120.

\bibitem[\protect\citeauthoryear{Baliga, Hanany, and Klibanoff}{Baliga
  et~al.}{2013}]{baliga2013}
\textsc{Baliga, S., E.~Hanany, and P.~Klibanoff} (2013): \enquote{Polarization
  and ambiguity,} \emph{American Economic Review}, 103, 3071--83.

\bibitem[\protect\citeauthoryear{Banerjee, Chandrasekhar, Duflo, and
  Jackson}{Banerjee et~al.}{2014}]{bcdj2017}
\textsc{Banerjee, A., A.~G. Chandrasekhar, E.~Duflo, and M.~O. Jackson} (2014):
  \enquote{Gossip: Identifying central individuals in a social network,} Tech.
  rep., National Bureau of Economic Research.

\bibitem[\protect\citeauthoryear{Banerjee and Fudenberg}{Banerjee and
  Fudenberg}{2004}]{bf2004}
\textsc{Banerjee, A. and D.~Fudenberg} (2004): \enquote{Word-of-mouth
  learning,} \emph{Games and economic behavior}, 46, 1--22.

\bibitem[\protect\citeauthoryear{Banerjee}{Banerjee}{1992}]{Banerjee1992}
\textsc{Banerjee, A.~V.} (1992): \enquote{A simple model of herd behavior,}
  \emph{The quarterly journal of economics}, 107, 797--817.

\bibitem[\protect\citeauthoryear{Banerjee}{Banerjee}{1993}]{Banerjee1993}
---\hspace{-.1pt}---\hspace{-.1pt}--- (1993): \enquote{The economics of
  rumours,} \emph{The Review of Economic Studies}, 60, 309--327.

\bibitem[\protect\citeauthoryear{Bowen, Dmitriev, and Galperti}{Bowen
  et~al.}{2021}]{bowen2021learning}
\textsc{Bowen, R., D.~Dmitriev, and S.~Galperti} (2021): \enquote{Learning from
  shared news: when abundant information leads to belief polarization,} Tech.
  rep., National Bureau of Economic Research.

\bibitem[\protect\citeauthoryear{Buechel, Kl{\"o}{\ss}ner, Meng, and
  Nassar}{Buechel et~al.}{2022}]{buechel2022misinformation}
\textsc{Buechel, B., S.~Kl{\"o}{\ss}ner, F.~Meng, and A.~Nassar} (2022):
  \enquote{Misinformation due to asymmetric information sharing,}
  \emph{Université de Fribourg Working Papers SES, N.528, Vl. 2022}.

\bibitem[\protect\citeauthoryear{Dandekar, Goel, and Lee}{Dandekar
  et~al.}{2013}]{Dandekar2013}
\textsc{Dandekar, P., A.~Goel, and D.~T. Lee} (2013): \enquote{Biased
  assimilation, homophily, and the dynamics of polarization,} \emph{Proceedings
  of the National Academy of Sciences}, 110, 5791--5796.

\bibitem[\protect\citeauthoryear{DeGroot}{DeGroot}{1974}]{DeGroot1974}
\textsc{DeGroot, M.~H.} (1974): \enquote{Reaching a consensus,} \emph{Journal
  of the American Statistical association}, 69, 118--121.

\bibitem[\protect\citeauthoryear{DeGroot and Schervish}{DeGroot and
  Schervish}{2012}]{DeGroot2012}
\textsc{DeGroot, M.~H. and M.~J. Schervish} (2012): \emph{Probability and
  statistics}, Pearson Education.

\bibitem[\protect\citeauthoryear{DeMarzo, Vayanos, and Zwiebel}{DeMarzo
  et~al.}{2003}]{DeMarzo2003}
\textsc{DeMarzo, P.~M., D.~Vayanos, and J.~Zwiebel} (2003): \enquote{Persuasion
  bias, social influence, and unidimensional opinions,} \emph{The Quarterly
  journal of economics}, 118, 909--968.

\bibitem[\protect\citeauthoryear{Ellison and Fudenberg}{Ellison and
  Fudenberg}{1993}]{Ellison1993}
\textsc{Ellison, G. and D.~Fudenberg} (1993): \enquote{Rules of thumb for
  social learning,} \emph{Journal of political Economy}, 101, 612--643.

\bibitem[\protect\citeauthoryear{Ellsberg}{Ellsberg}{1961}]{ellsberg1961}
\textsc{Ellsberg, D.} (1961): \enquote{Risk, ambiguity, and the Savage axioms,}
  \emph{The quarterly journal of economics}, 643--669.

\bibitem[\protect\citeauthoryear{Epstein, Noor, Sandroni et~al.}{Epstein
  et~al.}{2010}]{ens2010}
\textsc{Epstein, L.~G., J.~Noor, A.~Sandroni, et~al.} (2010):
  \enquote{Non-bayesian learning,} \emph{The BE Journal of Theoretical
  Economics}, 10, 1--20.

\bibitem[\protect\citeauthoryear{Epstein and Schneider}{Epstein and
  Schneider}{2007}]{Epstein2007}
\textsc{Epstein, L.~G. and M.~Schneider} (2007): \enquote{Learning under
  ambiguity,} \emph{The Review of Economic Studies}, 74, 1275--1303.

\bibitem[\protect\citeauthoryear{Fryer, Jackson et~al.}{Fryer
  et~al.}{2008}]{Fryer2008}
\textsc{Fryer, R., M.~O. Jackson, et~al.} (2008): \enquote{A categorical model
  of cognition and biased decision-making,} \emph{BE Journal of Theoretical
  Economics}, 8, 1--42.

\bibitem[\protect\citeauthoryear{Fryer~Jr, Harms, and Jackson}{Fryer~Jr
  et~al.}{2019}]{fhj2018}
\textsc{Fryer~Jr, R.~G., P.~Harms, and M.~O. Jackson} (2019): \enquote{Updating
  beliefs when evidence is open to interpretation: Implications for bias and
  polarization,} \emph{Journal of the European Economic Association}, 17,
  1470--1501.

\bibitem[\protect\citeauthoryear{Furnham and Marks}{Furnham and
  Marks}{2013}]{furnham2013}
\textsc{Furnham, A. and J.~Marks} (2013): \enquote{Tolerance of ambiguity: A
  review of the recent literature,} \emph{Psychology}, 4, 717--728.

\bibitem[\protect\citeauthoryear{Furnham and Ribchester}{Furnham and
  Ribchester}{1995}]{Furnham1995}
\textsc{Furnham, A. and T.~Ribchester} (1995): \enquote{Tolerance of ambiguity:
  A review of the concept, its measurement and applications,} \emph{Current
  psychology}, 14, 179--199.

\bibitem[\protect\citeauthoryear{Gale and Kariv}{Gale and
  Kariv}{2003}]{Gale2003}
\textsc{Gale, D. and S.~Kariv} (2003): \enquote{Bayesian learning in social
  networks,} \emph{Games and economic behavior}, 45, 329--346.

\bibitem[\protect\citeauthoryear{Gallo and Langtry}{Gallo and
  Langtry}{2020}]{gallo2020social}
\textsc{Gallo, E. and A.~Langtry} (2020): \enquote{Social networks,
  confirmation bias and shock elections,} \emph{arXiv preprint
  arXiv:2011.00520}.

\bibitem[\protect\citeauthoryear{Gennaioli and Shleifer}{Gennaioli and
  Shleifer}{2010}]{Gennaioli2010}
\textsc{Gennaioli, N. and A.~Shleifer} (2010): \enquote{What comes to mind,}
  \emph{The Quarterly journal of economics}, 125, 1399--1433.

\bibitem[\protect\citeauthoryear{Gilboa and Schmeidler}{Gilboa and
  Schmeidler}{1989}]{Gilboa1989}
\textsc{Gilboa, I. and D.~Schmeidler} (1989): \enquote{Maxmin expected utility
  with non-unique prior,} \emph{Journal of Mathematical Economics}, 18,
  141--153.

\bibitem[\protect\citeauthoryear{Gilboa and Schmeidler}{Gilboa and
  Schmeidler}{1993}]{Gilboa1993}
---\hspace{-.1pt}---\hspace{-.1pt}--- (1993): \enquote{Updating ambiguous
  beliefs,} \emph{Journal of economic theory}, 59, 33--49.

\bibitem[\protect\citeauthoryear{Glaeser and Sunstein}{Glaeser and
  Sunstein}{2013}]{Glaeser2013}
\textsc{Glaeser, E.~L. and C.~R. Sunstein} (2013): \enquote{Why does balanced
  news produce unbalanced views?} Tech. rep., National Bureau of Economic
  Research.

\bibitem[\protect\citeauthoryear{Golub and Jackson}{Golub and
  Jackson}{2010}]{Golub2010}
\textsc{Golub, B. and M.~O. Jackson} (2010): \enquote{Naive learning in social
  networks and the wisdom of crowds,} \emph{American Economic Journal:
  Microeconomics}, 2, 112--49.

\bibitem[\protect\citeauthoryear{Golub and Sadler}{Golub and
  Sadler}{2017}]{golub2017learning}
\textsc{Golub, B. and E.~Sadler} (2017): \enquote{Learning in social networks,}
  \emph{Available at SSRN 2919146}.

\bibitem[\protect\citeauthoryear{Grabisch and Rusinowska}{Grabisch and
  Rusinowska}{2020}]{grabisch2020survey}
\textsc{Grabisch, M. and A.~Rusinowska} (2020): \enquote{A survey on
  nonstrategic models of opinion dynamics,} \emph{Games}, 11, 65.

\bibitem[\protect\citeauthoryear{Han, Zikmund-Fisher, Duarte, Knaus, Black,
  Scherer, and Fagerlin}{Han et~al.}{2018}]{han2018communication}
\textsc{Han, P.~K., B.~J. Zikmund-Fisher, C.~W. Duarte, M.~Knaus, A.~Black,
  A.~M. Scherer, and A.~Fagerlin} (2018): \enquote{Communication of scientific
  uncertainty about a novel pandemic health threat: Ambiguity aversion and its
  mechanisms,} \emph{Journal of health communication}, 23, 435--444.

\bibitem[\protect\citeauthoryear{Hegselmann and Krause}{Hegselmann and
  Krause}{2005}]{HegKrau2005}
\textsc{Hegselmann, R. and U.~Krause} (2005): \enquote{Opinion dynamics driven
  by various ways of averaging,} \emph{Computational Economics}, 25, 381--405.

\bibitem[\protect\citeauthoryear{Hegselmann, Krause et~al.}{Hegselmann
  et~al.}{2002}]{HegKrau2002}
\textsc{Hegselmann, R., U.~Krause, et~al.} (2002): \enquote{Opinion dynamics
  and bounded confidence models, analysis, and simulation,} \emph{Journal of
  artificial societies and social simulation}, 5.

\bibitem[\protect\citeauthoryear{Hellman and Cover}{Hellman and
  Cover}{1970}]{Hellman1970}
\textsc{Hellman, M.~E. and T.~M. Cover} (1970): \enquote{{Learning with Finite
  Memory},} \emph{The Annals of Mathematical Statistics}, 41, 765 -- 782.

\bibitem[\protect\citeauthoryear{Jackson, Kalai, and Smorodinsky}{Jackson
  et~al.}{1999}]{JKS1999}
\textsc{Jackson, M.~O., E.~Kalai, and R.~Smorodinsky} (1999): \enquote{Bayesian
  representation of stochastic processes under learning: de Finetti revisited,}
  \emph{Econometrica}, 67, 875--893.

\bibitem[\protect\citeauthoryear{Jadbabaie, Molavi, Sandroni, and
  Tahbaz-Salehi}{Jadbabaie et~al.}{2012}]{jmst2012}
\textsc{Jadbabaie, A., P.~Molavi, A.~Sandroni, and A.~Tahbaz-Salehi} (2012):
  \enquote{Non-Bayesian social learning,} \emph{Games and Economic Behavior},
  76, 210--225.

\bibitem[\protect\citeauthoryear{Kalai and Lehrer}{Kalai and
  Lehrer}{1994}]{Kalai1994}
\textsc{Kalai, E. and E.~Lehrer} (1994): \enquote{Weak and strong merging of
  opinions,} \emph{Journal of Mathematical Economics}, 23, 73--86.

\bibitem[\protect\citeauthoryear{Lord, Ross, and Lepper}{Lord
  et~al.}{1979}]{Lord1979}
\textsc{Lord, C.~G., L.~Ross, and M.~R. Lepper} (1979): \enquote{Biased
  assimilation and attitude polarization: The effects of prior theories on
  subsequently considered evidence.} \emph{Journal of personality and social
  psychology}, 37, 2098.

\bibitem[\protect\citeauthoryear{Mahler}{Mahler}{1995}]{Mahler1995}
\textsc{Mahler, R.~P.} (1995): \enquote{Combining ambiguous evidence with
  respect to ambiguous a priori knowledge. Part II: Fuzzy Logic,} \emph{Fuzzy
  Sets and Systems}, 75, 319--354.

\bibitem[\protect\citeauthoryear{Mao, Bolouki, and Akyol}{Mao
  et~al.}{2018}]{Mao2018}
\textsc{Mao, Y., S.~Bolouki, and E.~Akyol} (2018): \enquote{Spread of
  information with confirmation bias in cyber-social networks,} \emph{IEEE
  Transactions on Network Science and Engineering}, 7, 688--700.

\bibitem[\protect\citeauthoryear{Mercier and Sperber}{Mercier and
  Sperber}{2011}]{Mercier2011}
\textsc{Mercier, H. and D.~Sperber} (2011): \enquote{Why do humans reason?
  Arguments for an argumentative theory.} \emph{Behavioral and brain sciences},
  34, 57--74.

\bibitem[\protect\citeauthoryear{Meyer}{Meyer}{2000}]{meyer2000}
\textsc{Meyer, C.~D.} (2000): \emph{Matrix analysis and applied linear
  algebra}, vol.~71, Siam.

\bibitem[\protect\citeauthoryear{Molavi, Tahbaz-Salehi, and Jadbabaie}{Molavi
  et~al.}{2018}]{Molavi2018}
\textsc{Molavi, P., A.~Tahbaz-Salehi, and A.~Jadbabaie} (2018): \enquote{A
  theory of non-Bayesian social learning,} \emph{Econometrica}, 86, 445--490.

\bibitem[\protect\citeauthoryear{Molden and Higgins}{Molden and
  Higgins}{2004}]{Molden2004}
\textsc{Molden, D.~C. and E.~T. Higgins} (2004): \enquote{Categorization under
  uncertainty: Resolving vagueness and ambiguity with eager versus vigilant
  strategies,} \emph{Social Cognition}, 22, 248--277.

\bibitem[\protect\citeauthoryear{Molden and Higgins}{Molden and
  Higgins}{2008}]{molden2008}
---\hspace{-.1pt}---\hspace{-.1pt}--- (2008): \enquote{How preferences for
  eager versus vigilant judgment strategies affect self-serving conclusions,}
  \emph{Journal of Experimental Social Psychology}, 44, 1219--1228.

\bibitem[\protect\citeauthoryear{Mullainathan}{Mullainathan}{2002}]{Mullainathan2002}
\textsc{Mullainathan, S.} (2002): \enquote{A memory-based model of bounded
  rationality,} \emph{The Quarterly Journal of Economics}, 117, 735--774.

\bibitem[\protect\citeauthoryear{Nickerson}{Nickerson}{1998}]{Nickerson1998}
\textsc{Nickerson, R.~S.} (1998): \enquote{Confirmation bias: A ubiquitous
  phenomenon in many guises,} \emph{Review of general psychology}, 2, 175--220.

\bibitem[\protect\citeauthoryear{Rabin and Schrag}{Rabin and
  Schrag}{1999}]{Rabin1999}
\textsc{Rabin, M. and J.~L. Schrag} (1999): \enquote{First impressions matter:
  A model of confirmatory bias,} \emph{The quarterly journal of economics},
  114, 37--82.

\bibitem[\protect\citeauthoryear{Sherman and Cohen}{Sherman and
  Cohen}{2006}]{Sherman2006}
\textsc{Sherman, D.~K. and G.~L. Cohen} (2006): \enquote{The psychology of
  self-defense: Self-affirmation theory,} \emph{Advances in experimental social
  psychology}, 38, 183--242.

\bibitem[\protect\citeauthoryear{Siegrist}{Siegrist}{2021}]{siegrist2021}
\textsc{Siegrist, K.} (2021): \enquote{Probability, Mathematical Statistics,
  Stochastic Processes,} \emph{Available online at:
  \url{https://stats.libretexts.org/Bookshelves/Probability_Theory} (Accessed
  07.01.2022).}

\bibitem[\protect\citeauthoryear{Sikder, Smith, Vivo, and Livan}{Sikder
  et~al.}{2020}]{sikder2020minimalistic}
\textsc{Sikder, O., R.~E. Smith, P.~Vivo, and G.~Livan} (2020): \enquote{A
  minimalistic model of bias, polarization and misinformation in social
  networks,} \emph{Scientific reports}, 10, 1--11.

\bibitem[\protect\citeauthoryear{Simonovic and Taber}{Simonovic and
  Taber}{2022}]{simonovic2022psychological}
\textsc{Simonovic, N. and J.~M. Taber} (2022): \enquote{Psychological impact of
  ambiguous health messages about COVID-19,} \emph{Journal of Behavioral
  Medicine}, 45, 159--171.

\bibitem[\protect\citeauthoryear{Sinkey}{Sinkey}{2015}]{sinkey2015experts}
\textsc{Sinkey, M.} (2015): \enquote{How do experts update beliefs? Lessons
  from a non-market environment,} \emph{Journal of Behavioral and Experimental
  Economics}, 57, 55--63.

\bibitem[\protect\citeauthoryear{Wilson}{Wilson}{2014}]{Wilson2014}
\textsc{Wilson, A.} (2014): \enquote{Bounded memory and biases in information
  processing,} \emph{Econometrica}, 82, 2257--2294.

\end{thebibliography}
